\documentclass[sigconf, nonacm]{acmart}

\setcopyright{none}
\renewcommand\footnotetextcopyrightpermission[1]{}
\AtBeginDocument{\fancyhead[RE,RO]{Mang, Xiang, et al.}\fancyhead[LE,LO]{PLOP: Cost-Based Placement of Semantic Operators in Hybrid Query Plans}\fancyfoot[C]{\thepage}\fancyfoot[R]{\small\textit{arXiv preprint, 2026}}\setlength{\footskip}{24pt}}

\AtBeginDocument{%
  \providecommand\BibTeX{{%
    \normalfont B\kern-0.5em{\scshape i\kern-0.25em b}\kern-0.8em\TeX}}}

\usepackage[scaled=.8]{beramono}
\usepackage{booktabs}
\usepackage{color}
\usepackage{colortbl}
\usepackage{tabularray}
\usepackage{float}

\definecolor{sky}{RGB}{0,70,140}

\newcommand{\sky}[1]{}

\newif\ifcomments
\commentsfalse
\ifcomments
    \providecommand{\alvin}[1]{{\color{brown}{alvin: #1 }}}
\else
    \providecommand{\alvin}[1]{}
\fi

\usepackage{pgfplots,xcolor,tikz}
\usetikzlibrary{calc,decorations.pathreplacing}
\pgfplotsset{compat=1.18}

\definecolor{colNone}{RGB}{156,186,156}      
\definecolor{colCost}{HTML}{89CCEB}          
\definecolor{colPull}{HTML}{FFA07A}          
\definecolor{colThal}{RGB}{60,100,160}       
\definecolor{colFlock}{RGB}{160,60,60}       
\definecolor{colPalm}{RGB}{60,130,80}        

\pgfplotsset{
  FIGALL/.style={
    ybar=1pt,
    enlarge x limits=0.035,
    x=1.7cm,                    
    width=\textwidth,
    height=4.2cm,
    tick label style={font=\large},
    label style={font=\large},
    title style={font=\Large\bfseries},
    xticklabel style={font=\large, rotate=55, anchor=east},
    yticklabel style={font=\large, text width=2.5em, align=right},
    ylabel style={at={(axis description cs:-0.03,0.5)}},
    bar width=4pt,              
    ymajorgrids=true,
    grid style={draw=gray!30},
    xtick=data,
    minor x tick num=0,
    legend style={
      font=\large\bfseries,
      at={(0.5,1.02)},
      anchor=south,
      legend columns=6,
      column sep=12pt,
      row sep=4pt,
      draw=none,
      /tikz/every even column/.append style={column sep=12pt},
    },
  },
  FIGALLbot/.style={
    FIGALL,
    legend style={opacity=0, text opacity=0, draw=none, fill=none, at={(10,10)}},
  },
}

\usepackage{makecell}
\usepackage{soul}
\usepackage{bbm}
\usepackage{amsthm}
\usepackage{balance}
\usepackage{multirow}
\usepackage[ruled,linesnumbered,vlined]{algorithm2e}
\usepackage{tcolorbox}
\usepackage[noend]{algpseudocode}
\usepackage[flushleft]{threeparttable}
\usepackage{fbox}
\usepackage{xcolor}
\usepackage{wasysym}
\usepackage{booktabs}
\usepackage{verbatim}
\usepackage{pgfplots}
\usepackage{hyperref}
\usepackage{fontawesome}
\usepackage{tikz}
\usepackage{graphics}
\usepackage{multirow}
\usepackage{fancybox}
\usepackage{listings}
\usepackage{enumitem}
\usetikzlibrary{patterns}
\usepackage{amsmath, amsfonts}
\usepackage{graphicx}
\usepackage{subfigure}
\usepackage{textcomp}
\usepackage{xcolor}
\usepackage{cleveref}

\theoremstyle{definition}
\newtheorem{problem}{Problem}
\theoremstyle{plain}

\newtheorem{mythm}{Theorem}[section]

\newcommand{\SF}{\mathsf{SF}}

\newcommand{\SJ}{\mathsf{SJ}}
\newcommand{\SP}{\mathsf{SP}}
\definecolor{add}{rgb}{0.2,0.4,0.6}

\crefname{section}{Sec.}{Secs.}
\Crefname{section}{Sec.}{Secs.}
\crefname{figure}{Fig.}{Figs.}
\Crefname{figure}{Fig.}{Figs.}
\crefname{algorithm}{Alg.}{Algs.}
\Crefname{algorithm}{Alg.}{Algs.}

\begin{document}

\title{PLOP: Cost-Based Placement of Semantic Operators in Hybrid Query Plans}

\renewcommand{\shortauthors}{Mang, Xiang, et al.}
\renewcommand\authorsaddresses{}

\author{Qiuyang Mang\textsuperscript{1,*} \quad
Yufan Xiang\textsuperscript{2,*} \quad
Hangrui Zhou\textsuperscript{1} \quad
Runyuan He\textsuperscript{1} \quad
Jiaxiang Yu\textsuperscript{1} \quad
Hanchen Li\textsuperscript{1} \quad
Aditya Parameswaran\textsuperscript{1} \quad
Alvin Cheung\textsuperscript{1}}
\affiliation{\institution{\textsuperscript{1}UC Berkeley \quad \textsuperscript{2}University of Wisconsin-Madison}\country{}}

\titlenote{*Equal contribution.}




\newcommand{\todo}[1]{}
\newcommand{\lhc}[1]{}
\newcommand{\jerry}[1]{}

\begin{abstract}
Recent database systems have introduced semantic operators that leverage large language models (LLMs) to filter, join, and project over structured data using natural language predicates.
These operators are often combined with traditional relational operators to produce hybrid query plans whose execution cost depends on both expensive LLM calls and conventional database processing.
A key optimization question is where to place each semantic operator relative to the relational operators in the plan: placing them early reduces the data that subsequent operators process, but requires more LLM calls; placing them late reduces LLM calls through deduplication, but forces relational operators to process larger intermediate data.
Existing systems either ignore this placement question or apply simple heuristics without considering the full cost trade-off.

We present \textsc{PLOP}, a plan-level optimizer for hybrid semantic-relational queries.
\textsc{PLOP} reduces hybrid query planning to semantic filter placement via two equivalence-preserving rewrites.
We prove that deferring all semantic filters to the latest possible position minimizes LLM invocations under function caching, but this heuristic can cause relational processing cost to dominate on complex multi-table queries.
To balance LLM cost against relational cost, \textsc{PLOP} uses a dynamic-programming-based cost model that finds the placement of semantic operators minimizing their weighted sum.
On 44 semantic SQL queries across five schemas and two benchmarks, \textsc{PLOP} achieves up to $1.5\times$ speedup and $4.29\times$ cost reduction while maintaining an average F1 of $0.85$ against an unoptimized baseline and $0.84$ against human-annotated ground truth.
Overall, \textsc{PLOP} achieves significant cost reduction while preserving the highest accuracy among six publicly available systems with similar functioning.
\end{abstract}

\maketitle

\section{Introduction}
\label{sec:intro}

Recent advances in large language models (LLMs) have opened up new opportunities for data analysis.
In particular, \emph{semantic queries} have gained growing attention, as they allow users to express data operations in natural language.
For instance, users can filter records with a predicate such as \emph{whether a movie review is positive}, or join two tables on a condition such as \emph{whether a customer is likely to buy a product}.
These semantic operations are powered by LLM calls and can be freely composed with traditional relational operators such as filters, joins, projections, and aggregations, forming \emph{hybrid query plans}. 

\begin{figure}[t]
\vspace{10pt}
\begin{lstlisting}[
  language=SQL,
  basicstyle=\ttfamily\small,
  frame=tb,
  framerule=0.4pt,
  rulecolor=\color{black},
  breaklines=true,
  showstringspaces=false,
  morekeywords={SEMANTIC},
  keywordstyle=\bfseries,
  stringstyle=\itshape,
  numbers=left,
  numberstyle=\tiny\color{gray},
  numbersep=6pt,
  xleftmargin=14pt,
  xrightmargin=2pt,
  aboveskip=4pt,
  belowskip=4pt,
  captionpos=b,
  abovecaptionskip=8pt,
  caption={Motivating semantic SQL query over a book review database.},
  label=lst:motivating,
  escapeinside={(*}{*)}
]
SELECT b.title, r.text
FROM books b
  JOIN reviews r ON b.book_id = r.book_id
WHERE SEMANTIC('{b.description} is about AI?') (*\hspace{6ex}//$\phi_1$*)
  AND SEMANTIC('{r.text} is a positive review?') (*\hspace{3.5ex}//$\phi_2$*)
  AND r.rating >= 3;
\end{lstlisting}
\vspace{-0.2in}
\end{figure}

As a motivating example, consider the query in Listing~\ref{lst:motivating}, which finds books about artificial intelligence that received positive reviews.
The query contains two semantic filters, $\SF_{\phi_1}$ on books and $\SF_{\phi_2}$ on reviews, an equi-join on \texttt{book\_id}, and a relational filter \texttt{rating} $\geq 3$.
Although semantic queries are powerful when combined with traditional relational operators, each semantic operator can require many LLM calls, making execution optimization crucial for both cost and latency.

\begin{figure*}[t]
\centering
\textit{``Find books about AI that received positive reviews with rating $\geq$ 3.''}
\vspace{4pt}

\resizebox{0.95\linewidth}{!}{
\begin{tikzpicture}[
  op/.style={draw, rounded corners=2pt, font=\small, inner sep=2pt, minimum width=1.1cm, align=center},
  sem/.style={op, fill=colCost!25},
  tbl/.style={draw, fill=gray!12, font=\small, inner sep=2pt, align=center},
  lbl/.style={font=\small\bfseries},
  ann/.style={font=\footnotesize, align=center},
  every edge/.style={draw, -latex, thick},
]

\node[lbl] at (0, 2.9) {(a) Push-down};
\node[op]  (A-join) at (0, 2.2) {$\Join_{\texttt{book\_id}}$};
\node[sem] (A-sf1)  at (-1.3, 1.3) {$\SF_{\phi_1}$};
\node[sem] (A-sf2)  at (1.3, 1.3)  {$\SF_{\phi_2}$};
\node[tbl] (A-bk)   at (-1.3, 0.5) {\texttt{books}};
\node[op]  (A-sig)  at (1.3, 0.5)  {$\sigma_{\texttt{rat}\geq 3}$};
\node[tbl] (A-rv)   at (1.3, -0.2) {\texttt{reviews}};
\node[ann] at (-1.3, 0.1) {\textit{1K rows}};
\node[ann] at (1.3, -0.55) {\textit{5K rows}};

\draw[-latex] (A-bk)  -- (A-sf1);
\draw[-latex] (A-sf1) -- (A-join);
\draw[-latex] (A-rv)  -- (A-sig);
\draw[-latex] (A-sig) -- (A-sf2);
\draw[-latex] (A-sf2) -- (A-join);

\node[ann] at (0, -0.95) {4{,}000 LLM calls\\[-1pt] Small join input};

\node[lbl] at (5.2, 2.9) {(b) Pull-up};
\node[sem] (B-sf1)  at (5.2, 2.2) {$\SF_{\phi_1}$ \tiny[cached]};
\node[sem] (B-sf2)  at (5.2, 1.4) {$\SF_{\phi_2}$ \tiny[cached]};
\node[op]  (B-join) at (5.2, 0.6) {$\Join_{\texttt{book\_id}}$};
\node[tbl] (B-bk)   at (3.9, -0.1) {\texttt{books}};
\node[op]  (B-sig)  at (6.5, -0.1) {$\sigma_{\texttt{rat}\geq 3}$};
\node[tbl] (B-rv)   at (6.5, -0.8) {\texttt{reviews}};
\node[ann] at (3.9, -0.45) {\textit{1K rows}};
\node[ann] at (6.5, -1.15) {\textit{5K rows}};

\draw[-latex] (B-bk)   -- (B-join);
\draw[-latex] (B-rv)   -- (B-sig);
\draw[-latex] (B-sig)  -- (B-join);
\draw[-latex] (B-join)  -- (B-sf2);
\draw[-latex] (B-sf2)  -- (B-sf1);

\node[ann] at (5.2, -1.55) {3{,}300 LLM calls\\[-1pt] Large join input};

\node[lbl] at (10.5, 2.9) {(c) \textsc{PLOP}};
\node[sem] (C-sf2)  at (10.5, 2.2) {$\SF_{\phi_2}$ \tiny[cached]};
\node[op]  (C-join) at (10.5, 1.4) {$\Join_{\texttt{book\_id}}$};
\node[sem] (C-sf1)  at (9.2, 0.5) {$\SF_{\phi_1}$};
\node[op]  (C-sig)  at (11.8, 0.5) {$\sigma_{\texttt{rat}\geq 3}$};
\node[tbl] (C-bk)   at (9.2, -0.2) {\texttt{books}};
\node[tbl] (C-rv)   at (11.8, -0.2) {\texttt{reviews}};
\node[ann] at (9.2, -0.55) {\textit{1K rows}};
\node[ann] at (11.8, -0.55) {\textit{5K rows}};

\draw[-latex] (C-bk)   -- (C-sf1);
\draw[-latex] (C-sf1)  -- (C-join);
\draw[-latex] (C-rv)   -- (C-sig);
\draw[-latex] (C-sig)  -- (C-join);
\draw[-latex] (C-join)  -- (C-sf2);

\node[ann] at (10.5, -0.95) {Balanced cost};

\end{tikzpicture}
}
\caption{Three hybrid query plans for Listing~\ref{lst:motivating} (\texttt{books}: 1K rows, \texttt{reviews}: 5K rows).
(a)~\textbf{Push-down}: filters applied before the join. Small join, but 4K LLM calls.
(b)~\textbf{Pull-up}: filters deferred after the join. Function caching evaluates each distinct tuple once, yielding only 3.3K calls, but the join processes unfiltered tables.
(c)~\textbf{\textsc{PLOP}}: $\phi_1$ pushed down to shrink the join, $\phi_2$ pulled up to reduce LLM calls.
}
\Description{Three hybrid query plan trees showing push-down, pull-up, and PLOP placement strategies.}
\label{fig:plans}
\vspace{-0.2in}
\end{figure*}
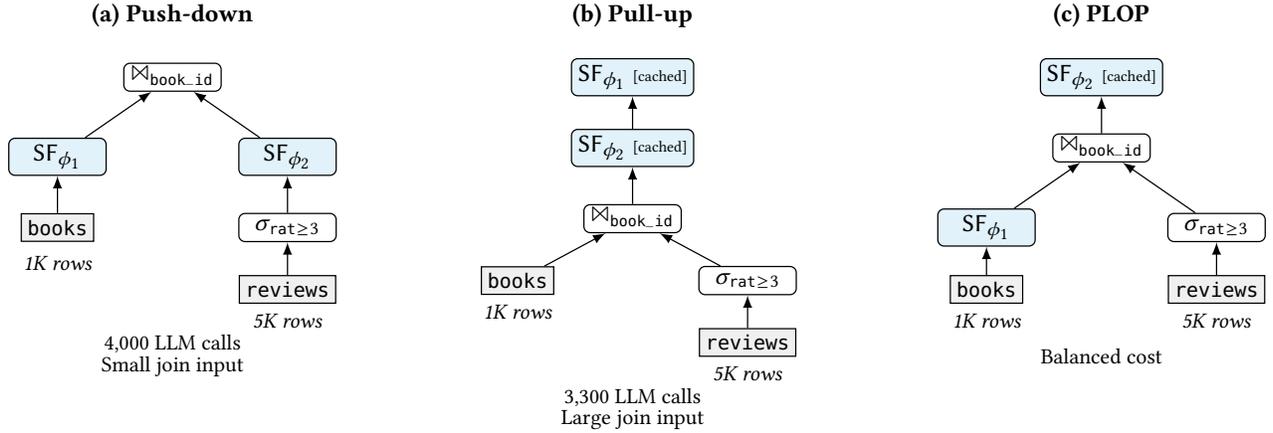

Several systems have been proposed to reduce costs and optimize semantic query execution.
Lotus~\cite{patel2025semantic} uses model cascades~\cite{kang2020approximate}, where a proxy model handles easy cases, and defers to a model for harder ones.
Palimpzest~\cite{liu2025palimpzest} and Abacus~\cite{russo2025abacuscostbasedoptimizersemantic} search for optimal physical implementations of individual operators using a cost model.
DocETL~\cite{shankar2024docetl} uses semantically equivalent rewrites to optimize queries, while ZenDB~\cite{lin2025zendb} uses document structure to reduce query plan costs.
Other systems extend SQL with semantic operators and contribute their own optimization strategies, including FlockMTL~\cite{dorbani2025beyond}, ThalamusDB~\cite{jo2024thalamusdb}, iPDB~\cite{kumarasinghe2026ipdb}, Sema~\cite{qi2026sema}, and Cortex AISQL~\cite{liskowski2025cortex}.
While prior work primarily focuses on optimizing semantic operators, query plans involve both semantic and relational operators, requiring joint optimization and careful consideration of accuracy due to the following two key limitations.

\begin{itemize}[leftmargin=*, itemsep=0.5ex]
    \item \textbf{Ignoring semantic-relational interactions.} Existing systems optimize semantic operators in isolation without considering how their placement interacts with relational operators. As illustrated in \Cref{fig:plans}, where a semantic filter is placed relative to joins and relational filters can dramatically affect both LLM cost and relational execution cost. ZenDB~\cite{lin2025zendb}, iPDB~\cite{kumarasinghe2026ipdb}, and Cortex AISQL~\cite{liskowski2025cortex} apply pull-up of semantic filters as a heuristic, but pulling up is not always beneficial: it can inflate intermediate results and cause relational costs to explode, especially in multi-table queries~\cite{liskowski2025cortex}. No existing system reasons about the trade-off between LLM cost and relational cost when deciding where to place semantic operators.
    \item \textbf{Accuracy risk from approximation.} Most existing systems reduce LLM cost by approximating or modifying individual operators. Model cascading routes easy cases to a cheaper proxy model, and approximate query processing replaces exact LLM evaluation with sampling-based estimates. While effective in isolation, these approximations alter operator outputs and can compound when multiple operators are composed in the same plan, causing significant deviations in the final result.
\end{itemize}

\noindent This motivates the problem we address: \emph{given a hybrid query plan that interleaves semantic and relational operators, where should each semantic operator be placed to minimize the combined LLM and relational execution cost, without changing the semantics or implementation of individual operators?} 

To address this problem, we present \textsc{PLOP}, a plan-level optimization framework for hybrid semantic-relational queries that optimizes the \emph{placement} of semantic operators within the plan tree.
We begin by reducing the general problem of semantic operator placement to semantic filter placement via two equivalence-preserving reductions. 
By doing so, all remaining semantic operators are semantic filters, and the reduced problem can be optimized based on two key insights:

\noindent \textbf{(1)} We show that pulling up all semantic filters above relational operators, combined with function caching~\cite{hellerstein1993predicate} (i.e., caching semantic operator calls so that each distinct prompt is evaluated only once by LLM), minimizes the total number of LLM invocations.
Consider the example in Listing~\ref{lst:motivating} and the three plans in \Cref{fig:plans}.
Suppose \texttt{books} has 1{,}000 rows and \texttt{reviews} has 5{,}000 rows.
The push-down plan (\Cref{fig:plans}a) evaluates $\SF_{\phi_1}$ on all 1{,}000 books and $\SF_{\phi_2}$ on the 3{,}000 reviews that pass the relational filter, totaling 4{,}000 LLM calls.
In the pull-up plan shown in \Cref{fig:plans}b, the relational filter retains 3{,}000 reviews with \texttt{rating} $\geq 3$, and the equi-join matches them with books.
Assuming the join is selective, and that only 800 of the 1{,}000 books appear in the join output, along with 2{,}500 distinct reviews.
By evaluating each distinct prompt once in LLM, yielding only $800 + 2{,}500 = 3{,}300$ LLM calls instead of 4{,}000. 

\noindent \textbf{(2)} We observe that this strategy of minimizing LLM calls is not always globally optimal.
For example, in production workloads at Snowflake, nearly 40\% of semantic queries involve multiple tables~\cite{liskowski2025cortex}.
Queries with 5+ joins and multiple semantic filters are common in analytical scenarios.
Consider a chain of five equi-joins over five tables, each with 1{,}000 rows and each carrying its own semantic filter with 10\% selectivity.
Pushing down all filters first shrinks each table to about 100 rows, so the subsequent joins operate on small inputs.
Pulling up all filters instead forces each join to process the full 1{,}000-row tables.
Since intermediate sizes compound across multiple joins, the gap between the two strategies widens rapidly with the number of tables.
Although pulling up still reduces LLM calls as shown above, the relational cost of joining unfiltered tables quickly dominates.
To navigate this trade-off, \textsc{PLOP} employs a dynamic programming (DP)-based cost model that jointly considers LLM and relational execution cost to determine the optimal placement of each semantic filter.

We evaluate \textsc{PLOP} on 44 semantic SQL queries spanning five database schemas from DataAgentBench~\cite{ma2026can}, TPC-H~\cite{tpch2024}, and SemBench~\cite{lao2026sembenchbenchmarksemanticquery}, comparing against six baseline methods.
\textsc{PLOP} achieves up to $1.5\times$ speedup and $4.3\times$ cost reduction over the DuckDB UDF~\cite{raasveldt2019duckdb} baseline while maintaining an average F1 score of $0.85$.
Competing systems either sacrifice accuracy for cost savings or achieve a limited reduction at moderate accuracy.
For example, ThalamusDB~\cite{jo2024thalamusdb} reduces cost by $10\times$ but at only with a F1 score of $0.52$
FlockMTL~\cite{dorbani2025beyond} achieves $7.3\times$ with a F1 score of 0.34.
Palimpzest~\cite{liu2025palimpzest} reaches $2.1\times$ cost reduction with a F1 score of $0.53$.
We also show that pull-up alone suffices for queries on 1 -- 2 tables, but can cause join costs to explode for complex multi-table queries.
Our ablation studies show that \textsc{PLOP}'s cost model can accommodate a wide range of relative costs between relational and semantic operators, while remaining robust to selectivity estimation errors and incurring only negligible optimization overhead.

\noindent In summary, we make the following contributions:
\vspace{-2ex}
\begin{itemize}[leftmargin=*, itemsep=0.5ex]
    \item We formalize the problem of optimizing hybrid query plans that interleave semantic and relational operators, and show how to reduce it to semantic filter placement by pulling up semantic projections and decomposing semantic joins.
    \item We prove that pulling up semantic filters with function caching minimizes LLM invocations, and propose a DP-based cost model that balances LLM and relational costs for cases where pull-up alone degrades overall performance.
    \item We benchmark 44 semantic SQL queries across five schemas and conduct an extensive evaluation against six baselines. The results show that \textsc{PLOP}'s cost model achieves the best trade-off between efficiency and query accuracy, reducing cost by up to 33$\times$ and latency by up to 3.0$\times$.
\end{itemize}

\section{Background}
\label{sec:background}

\subsection{Semantic Operators}
A semantic operator is a query operator parameterized by a natural language expression and evaluated by an LLM~\cite{patel2025semantic}.
We study three types of semantic operators: Semantic Filter ($\SF$), Semantic Join ($\SJ$), and Semantic Projection ($\SP$).
Let $R$ and $S$ be relations, $\phi$ a natural language expression, and $\mathcal{M}$ an LLM.

\noindent \textbf{Semantic Filter ($\SF$)} retains tuples in $R$ for which $\mathcal{M}$ evaluates predicate $\phi$ as true:
\[
\SF_{\phi}(R) = \{ r \in R \mid \mathcal{M}(r, \phi) = \text{true} \}.
\]
\textbf{Semantic Join ($\SJ$)} pairs tuples from $R$ and $S$ that satisfy a natural language condition $\phi$:
\[
\SJ_{\phi}(R, S) = \{ (r, s) \mid r \in R,\, s \in S,\, \mathcal{M}(r, s, \phi) = \text{true} \}.
\]
\textbf{Semantic Projection ($\SP$)} produces new attributes by transforming each tuple in $R$ according to $\phi$:
\[
\SP_{\phi}(R) = \{ \mathcal{M}(r, \phi) \mid r \in R \}.
\]

Under these definitions, $\SF$ and $\SP$ require one LLM call per input tuple, or per tuple pair for $\SJ$, making LLM invocation the dominant time cost in hybrid query execution.

\subsection{Hybrid Query Plans}
A hybrid query plan is a logical query plan that integrates both relational and semantic operators.
Formally, it is a rooted tree $\mathcal{T} = (V, E)$ where each node $v \in V$ is either a relational operator, such as filter $\sigma$, projection $\pi$, join $\Join$, or aggregation $\gamma$, or a semantic operator from $\{\SF, \SJ, \SP\}$.
Each node takes the output of its children as input and produces a relation as the output.
\Cref{fig:plans} shows three hybrid plans for the query in Listing~\ref{lst:motivating}.
Different placements of the same semantic operators lead to very different cost profiles: pushing down reduces join input but increases LLM calls, while pulling up reduces LLM calls but inflates the join.

\subsection{Function Caching}
\label{sec:function-caching}
\emph{Function caching} for expensive predicates was first proposed in the context of predicate migration~\cite{hellerstein1993predicate}.
When an expensive predicate depends only on columns from one side of a join, its result can be computed once per distinct input tuple and reused across all join-expanded rows.

We adopt this technique for semantic operators.
When a semantic filter $\SF_\phi$ depends only on columns from a base table $A$ and is placed above a join $A \Join B$, the same tuple $a$ may appear in multiple output rows paired with different tuples from $B$.
With function caching, the LLM evaluates each distinct $a$ at most once and reuses the cached result.
The number of LLM calls is therefore equal to the number of distinct $A$-tuples in the join output, which can be much smaller than $|A|$ if the join eliminates unmatched tuples.

This property is critical for optimization: pulling up a semantic filter above a join can reduce the number of LLM calls by skipping tuples that the join eliminates.
With function caching, the effective cost of a semantic filter depends on the number of \emph{distinct} relevant tuples in the intermediate result rather than the total number of tuples.
These cost characteristics motivate a formal optimization framework.

\section{Problem Specification}
\label{sec:problem}
We now formally define the hybrid plan optimization problem, then show how to simplify it by reducing the general semantic operator placement problem to semantic filter placement only.
\subsection{Problem Definition}
\label{sec:problem-def}

We formalize the hybrid plan optimization problem as follows.

\begin{problem}[Hybrid Plan Optimization]
\label{prob:main}
Given a hybrid query plan tree $\mathcal{T}$ and a user-specified parameter $\alpha > 0$, find the plan $\mathcal{T}^*$ that minimizes:
\[
C(\mathcal{T}^*) = C_{\mathrm{LLM}}(\mathcal{T}^*) + \alpha \cdot C_{\mathrm{rel}}(\mathcal{T}^*),
\]
where $C_{\mathrm{LLM}}$ is the total number of distinct rows processed by semantic operators under function caching, and $C_{\mathrm{rel}}$ is the total number of rows processed by relational operators.
The search space consists of all plans obtainable by equivalent repositioning and rewriting semantic operators in $\mathcal{T}$ that produce the same output. 
\end{problem}

Both costs are measured by the number of rows, since each semantic operator requires one LLM call per input row under the definition.
Here, $\alpha$ converts between the per-row cost of each type of operator.
In practice, $\alpha$ reflects the relative weight assigned to CPU-based relational processing versus LLM inference in the overall execution cost, where cost may capture both latency and monetary expense: a small $\alpha$ prioritizes reducing LLM-related cost, while a large $\alpha$ prioritizes reducing relational processing cost.

\subsection{Problem Simplification}
\label{sec:simplification}

Solving Problem~\ref{prob:main} in full generality requires jointly optimizing join ordering and semantic operator placement.
In this work, we fix the join order as given by the underlying optimizer and focus on semantic operator placement. 
We further simplify the problem through two reductions. 

\paragraph{Pulling up semantic projections.}
The first reduction is designed for semantic projections. 
Pulling up a semantic projection to the highest feasible position in the plan tree reduces LLM cost: with function caching, Semantic Projection ($\SP$) is evaluated only on distinct tuples, and placing it higher means the operators below reduce the number of distinct tuples that reach it.
Since $\SP$ does not change cardinality and only adds columns, this rewriting is semantics-preserving provided no intervening operator references the new columns.
Meanwhile, a semantic projection $\SP_\phi$ generates new columns that may be referenced by downstream operators.
When a relational operator references a column produced by $\SP$, the operator depends on $\SP$'s output. 
Therefore, if $\SP$ is pulled up, the dependent relational operator must move with it to preserve this dependency.
We determine the final position via a topological sort that respects these column dependencies.
To illustrate, consider the query in Listing~\ref{lst:sp-example} and the illustration of its transformation in \Cref{fig:sp-pullup}.
This query computes a sentiment score via $\SP$ and then filters on it.
To pull up the $\SP$ above the join, we first pull up the filter $\sigma_{\texttt{score}\geq 4}$ that refers to it.

\begin{figure}[t]
\begin{lstlisting}[
  language=SQL, basicstyle=\ttfamily\small,
  frame=tb, framerule=0.4pt, rulecolor=\color{black},
  breaklines=true, showstringspaces=false,
  morekeywords={SEMANTIC_INT}, keywordstyle=\bfseries,
  numbers=left, numberstyle=\tiny\color{gray},
  numbersep=6pt, xleftmargin=14pt, xrightmargin=2pt,
  aboveskip=4pt, belowskip=4pt, captionpos=b, abovecaptionskip=8pt,
  caption={\vspace{2pt}Semantic projection with a relational filter.},
  label=lst:sp-example,
]
SELECT b.title,
  SEMANTIC_INT('Rate {r.text} sentiment 1-5') AS score
FROM books b JOIN reviews r ON b.id = r.book_id
WHERE score >= 4;
\end{lstlisting}
\Description{SQL query with SEMANTIC_INT projection and a filter on the projected column.}
\end{figure}

\begin{figure}[t]
\centering
\resizebox{0.85\columnwidth}{!}{%
\begin{tikzpicture}[
  op/.style={draw, rounded corners=2pt, font=\small, inner sep=3pt, minimum width=1.3cm, align=center},
  sem/.style={op, fill=colCost!25},
  tbl/.style={draw, fill=gray!12, font=\small, inner sep=3pt, align=center},
  lbl/.style={font=\normalsize\bfseries},
  ann/.style={font=\footnotesize, align=center},
  arr/.style={-latex, thick},
]

\node[lbl] at (-1, 3.8) {(a) Before};
\node[op]  (A-join) at (-1, 3.0) {$\Join$};
\node[op]  (A-sig)  at (-2.2, 2.0) {$\sigma_{\texttt{score}\geq 4}$};
\node[tbl] (A-bk)   at (0.2, 2.0) {\texttt{books}};
\node[sem] (A-sp)   at (-2.2, 1.0) {$\SP$: \texttt{score}};
\node[tbl] (A-rv)   at (-2.2, 0.0) {\texttt{reviews}};

\draw[arr] (A-rv)  -- (A-sp);
\draw[arr] (A-sp)  -- (A-sig);
\draw[arr] (A-sig) -- (A-join);
\draw[arr] (A-bk)  -- (A-join);

\node[font=\Huge] at (1.5, 1.8) {$\Rightarrow$};

\node[lbl] at (3.8, 3.8) {(b) After};
\node[op]  (B-sig)  at (3.8, 3.0) {$\sigma_{\texttt{score}\geq 4}$};
\node[sem] (B-sp)   at (3.8, 2.1) {$\SP$: \texttt{score}};
\node[op]  (B-join) at (3.8, 1.2) {$\Join$};
\node[tbl] (B-rv)   at (2.5, 0.2) {\texttt{reviews}};
\node[tbl] (B-bk)   at (5.1, 0.2) {\texttt{books}};

\draw[arr] (B-rv)   -- (B-join);
\draw[arr] (B-bk)   -- (B-join);
\draw[arr] (B-join)  -- (B-sp);
\draw[arr] (B-sp)   -- (B-sig);

\end{tikzpicture}%
}
\caption{Pulling up $\SP$ and its dependent $\sigma$. (a)~$\SP$ computes a sentiment score on all reviews before the join. (b)~After pull-up, $\SP$ evaluates only reviews that have a matching book.}
\Description{Before and after plan trees showing SP pull-up above a join.}
\label{fig:sp-pullup}
\end{figure}
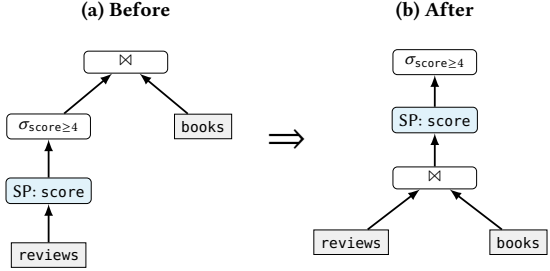

\paragraph{Decomposing semantic joins.}
The second reduction is for semantic joins.
From the definitions of $\SJ$ and $\SF$, we can rewrite $\SJ_\phi(R, S)$ as $\SF_\phi(R \times S)$.
\Cref{fig:sj-decomp} shows an example where a relational filter $\sigma$ coexists with a semantic join.
After decomposition, the semantic predicate becomes a separate $\SF$ that can be repositioned independently, just like any other semantic filter.
Note that the relational cost of the cross join is negligible compared to the $\SJ$ it replaces, since the original $\SJ$ already evaluates every pair with an LLM call.
Once decomposed, the $\SF$ can be repositioned by the optimizer, e.g., any relational filters can be pushed between $\times$ and $\SF$ to reduce the number of pairs before LLM evaluation. 
Note that prior work~\cite{patel2025semantic} defines semantic join only for inner joins, so our decomposition also considers only inner semantic joins.

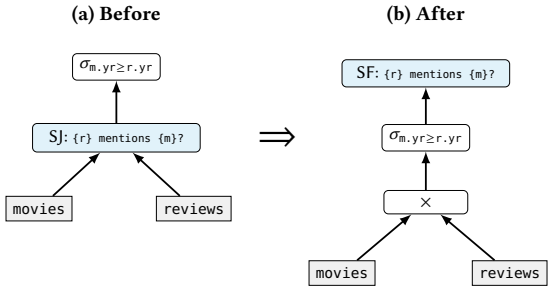
\begin{figure}[t]
\centering
\resizebox{0.85\columnwidth}{!}{%
\begin{tikzpicture}[
  op/.style={draw, rounded corners=2pt, font=\small, inner sep=3pt, minimum width=1.3cm, align=center},
  sem/.style={op, fill=colCost!25},
  tbl/.style={draw, fill=gray!12, font=\small, inner sep=3pt, align=center},
  lbl/.style={font=\normalsize\bfseries},
  arr/.style={-latex, thick},
]

\node[lbl] at (-1, 4.4) {(a) Before};
\node[op]  (A-sig) at (-1, 3.6) {$\sigma_{\texttt{m.yr} \geq \texttt{r.yr}}$};
\node[sem, minimum width=2.6cm] (A-sj) at (-1, 2.5) {$\SJ$: \scriptsize\texttt{\{r\} mentions \{m\}?}};
\node[tbl] (A-m)    at (-2.2, 1.4) {\texttt{movies}};
\node[tbl] (A-r)    at (0.2, 1.4) {\texttt{reviews}};

\draw[arr] (A-m)   -- (A-sj);
\draw[arr] (A-r)   -- (A-sj);
\draw[arr] (A-sj)  -- (A-sig);

\node[font=\Huge] at (1.5, 2.5) {$\Rightarrow$};

\node[lbl] at (3.8, 4.4) {(b) After};
\node[sem, minimum width=2.6cm] (B-sf) at (3.8, 3.5) {$\SF$: \scriptsize\texttt{\{r\} mentions \{m\}?}};
\node[op]  (B-sig)  at (3.8, 2.5) {$\sigma_{\texttt{m.yr} \geq \texttt{r.yr}}$};
\node[op]  (B-cross) at (3.8, 1.5) {$\times$};
\node[tbl] (B-m)    at (2.5, 0.4) {\texttt{movies}};
\node[tbl] (B-r)    at (5.1, 0.4) {\texttt{reviews}};

\draw[arr] (B-m)     -- (B-cross);
\draw[arr] (B-r)     -- (B-cross);
\draw[arr] (B-cross)  -- (B-sig);
\draw[arr] (B-sig)   -- (B-sf);

\end{tikzpicture}%
}
\caption{Decomposing a semantic join. (a)~$\sigma$ sits above the monolithic $\SJ$, which evaluates the semantic predicate on all pairs. (b)~After decomposition, $\sigma$ is placed between $\times$ and $\SF$, filtering pairs before LLM evaluation.}
\Description{Plan tree showing semantic join decomposed into cross join plus semantic filter with filter pushdown.}
\label{fig:sj-decomp}
\end{figure}

\paragraph{Reduced problem.}
We apply these two reductions repeatedly until no $\SJ$ or movable $\SP$ remains.
In particular, decomposing an $\SJ$ produces a new $\SF$, which can then be pulled up just like any other $\SF$.
Pulling up that $\SF$ may also enable pulling up an $\SP$ whose new column was previously used in the $\SJ$.
After convergence, all $\SP$ positions are fixed, and every $\SJ$ has been rewritten into $\SF$.
The remaining optimization concerns only $\SF$ placement: finding where to position each semantic filter within the fixed plan tree to minimize $C_{\mathrm{LLM}} + \alpha \cdot C_{\mathrm{rel}}$.
We discuss this next.

\section{Methodology}
\label{sec:method}

After applying the simplification in \Cref{sec:simplification}, the plan tree only has relational operators and semantic filters to consider: semantic projections have been fixed at their final highest positions, and semantic joins have been decomposed.
All results in this section apply to this simplified tree.
We now address the reduced problem of finding the optimal placement for each semantic filter.
We first show in \Cref{sec:pullup} that, with function caching, pulling up always minimizes LLM cost.
We then present in \Cref{sec:dp} a cost model that balances LLM and relational cost via dynamic programming.

\subsection{Pulling Up Semantic Filters}
\label{sec:pullup}

Recall from \Cref{sec:function-caching} that function caching evaluates each distinct input at most once.
As a result, the cost of a semantic filter is determined not by the number of tuples it processes, but by the number of \emph{distinct inputs} it receives.

Let $\SF_\phi$ reference a set of columns $C$.
When evaluated at a plan node $u$, the number of LLM calls equals the number of distinct non-null $\pi_C$-projections in the intermediate result at $u$, which we denote by $d_C(u)$.
We exclude null projections because, similar to SQL null semantics, we define them to require no LLM call, i.e., $\SF_\phi(\texttt{NULL}) = \texttt{NULL}$.
Thus, optimizing the placement of $\SF_\phi$ reduces to minimizing $d_C(u)$.

A natural question is whether placing $\SF_\phi$ higher in the plan always reduces its cost.
The following theorem shows that if $\SF_\phi$ is moved only across \emph{non-blocking operators} whose swap with $\SF$s preserves query semantics (i.e., relational filters, cross-products, projections, inner joins, and other $\SF$s), and not others (e.g., \texttt{LIMIT} and \texttt{UNION}), then the number of distinct inputs seen by $\SF_\phi$ can only decrease.

\begin{mythm}[Pull-up optimality]
\label{thm:pullup}
For any ancestor $u'$ of $u$ in the plan tree, if the path from $u$ to $u'$ contains only non-block operators, then
$d_C(u') \leq d_C(u)$.
\end{mythm}

\begin{proof}
Let $D_C(v)$ denote the set of distinct $\pi_C$-projections in the intermediate result at node $v$.
It suffices to show that each operator on the path preserves or shrinks $D_C$. 

A filter $\sigma$ removes tuples, so $D_C(\sigma(R)) \subseteq D_C(R)$.
A projection $\pi$ that retains all columns in $C$ preserves distinct values, i.e., $D_C(\pi(R)) = D_C(R)$.
An inner join may replicate tuples but does not introduce new values for columns in $C$, so $D_C$ cannot increase.
Applying these observations along the path from $u$ to $u'$ yields $d_C(u') \leq d_C(u)$.
\end{proof}


We can see that pulling up $\SF_\phi$ across non-block relational operators is semantics-preserving.
Since $\SF_\phi$ depends only on columns in $C$, it removes the same tuples regardless of where it is applied.

\paragraph{Pull-up algorithm.}
\begin{algorithm}[t]
\caption{Semantic filter pull-up.}
\label{alg:pullup}
\KwIn{Plan tree $\mathcal{T}$ with semantic filters $\SF_1, \ldots, \SF_n$}
\KwOut{Modified plan tree with all filters pulled up}
$\textit{changed} \leftarrow \texttt{true}$\;
\While{changed\nllabel{ln:while}}{
    $\textit{changed} \leftarrow \texttt{false}$\;
    \ForEach{$\SF_i$ in $\mathcal{T}$\nllabel{ln:foreach}}{
        $p \leftarrow \text{parent}(\SF_i)$\;
        \If{$p \neq \texttt{root}$ \textbf{and} $p$ is not a block operator \textbf{and} $p$ is not an $\SF$\nllabel{ln:check}}{
            \If{$p$ is a projection $\pi$\nllabel{ln:projif}}{
                add columns referenced by $\SF_i$ to $\pi$\nllabel{ln:addcol}\;
            }
            swap $\SF_i$ with $p$\nllabel{ln:swap}\;
            $\textit{changed} \leftarrow \texttt{true}$\nllabel{ln:mark}\;
        }
    }
}
\Return $\mathcal{T}$\;
\end{algorithm}

Based on \Cref{thm:pullup}, we greedily pull every semantic filter as high as possible.
\Cref{alg:pullup} describes the procedure.
The outer loop (Line~\ref{ln:while}) repeats until no swap occurs in a full pass.
In each pass, the algorithm scans every $\SF_i$ (Line~\ref{ln:foreach}) and checks whether its parent $p$ is swappable (Line~\ref{ln:check}): the swap is skipped if $p$ is the root, a block operator, or another $\SF$.
If $p$ is a projection, the columns referenced by $\SF_i$ are added to $\pi$ before the swap (Lines~\ref{ln:projif}--\ref{ln:addcol}) so the predicate remains evaluable.
The swap itself happens at Line~\ref{ln:swap}, and Line~\ref{ln:mark} marks progress for the outer loop.
The outer loop is necessary because moving one filter may unblock another.
For example, if $\SF_1$ sits below a join that is below $\SF_2$, pulling $\SF_2$ above a higher operator in one iteration makes the join the new parent of $\SF_1$, which can then be swapped in the next iteration.

\begin{mythm}[Complexity of \Cref{alg:pullup}]
\label{thm:pullup-complexity}
\Cref{alg:pullup} terminates in $\mathcal{O}(n^2 \cdot d)$ time, where $n$ is the number of semantic filters and $d$ is the depth of the plan tree.
\end{mythm}

\begin{proof}
Each semantic filter can be pulled up at most $d$ levels before reaching the root or an aggregation.
Since only swaps with non-$\SF$ operators count as progress, and there are $n$ filters, the outer loop executes at most $n \cdot d$ rounds.
Each round scans all $n$ filters, giving $\mathcal{O}(n^2 \cdot d)$ total work.
\end{proof}

\paragraph{Why pull-up alone is not sufficient.}
While \Cref{thm:pullup} shows that pull-up minimizes LLM cost, it can significantly increase relational cost (e.g., increasing join latency).
When all filters are pulled above the joins, the joins must process full, unfiltered tables, and intermediate result sizes can grow rapidly with the number of tables.
This motivates a cost model that balances the two objectives, which we present next.

\begin{figure*}[t!]
\centering
\resizebox{0.95\linewidth}{!}{%
\begin{tikzpicture}[
  op/.style={draw, rounded corners=2pt, font=\small, inner sep=3pt, minimum width=1.1cm, align=center},
  sem/.style={op, fill=colCost!25},
  tbl/.style={draw, fill=gray!12, font=\small, inner sep=3pt, align=center},
  box/.style={draw, rounded corners=4pt, font=\small, inner sep=6pt, minimum width=3cm, minimum height=1.3cm, align=center, line width=0.6pt},
  stepbox/.style={box, fill=gray!6, draw=gray!60},
  resbox/.style={draw, rounded corners=4pt, font=\small, inner sep=5pt, align=center, fill=colCost!12, draw=colCost!60, line width=0.6pt},
  winbox/.style={resbox, draw=colCost!80, line width=1.2pt, fill=colCost!20},
  arr/.style={-latex, thick, color=gray!70},
  treearr/.style={-latex},
  flowarr/.style={-latex, thick, dashed, color=colCost},
  lbl/.style={font=\normalsize\bfseries, color=black},
  sublbl/.style={font=\scriptsize, color=gray!60!black},
  ann/.style={font=\scriptsize, color=gray!50!black, align=left},
]

\node[lbl] at (-3.2, 3.4) {Input};
\node[op]  (T-join) at (-3.2, 2.5) {$u$: $\Join$};
\node[sem] (T-sf1) at (-4.5, 1.5) {$\SF_1$\,\scriptsize($s_1$)};
\node[sem] (T-sf2) at (-1.9, 1.5) {$\SF_2$\,\scriptsize($s_2$)};
\node[tbl] (T-bk)   at (-4.5, 0.5) {$L$\;\scriptsize($|L|$)};
\node[tbl] (T-rv)   at (-1.9, 0.5) {$R$\;\scriptsize($|R|$)};

\draw[treearr] (T-bk) -- (T-sf1);
\draw[treearr] (T-sf1) -- (T-join);
\draw[treearr] (T-rv) -- (T-sf2);
\draw[treearr] (T-sf2) -- (T-join);

\draw[gray!25, line width=0.4pt] (-0.6, -0.1) -- (-0.6, 3.7);

\node[lbl] at (2.0, 3.4) {\textbf{(1)} Distribute};
\node[stepbox, minimum width=3.2cm] (s1) at (2.0, 2.3) {
  $dp_{L,\{1\}} + dp_{R,\emptyset}$
};
\node[resbox] (r1) at (2.0, 0.7) {$dp_{u,\{1\}}$};
\draw[arr] (s1) -- (r1);

\node[lbl] at (6.2, 3.4) {\textbf{(2)} Relational cost};
\node[stepbox, minimum width=3.2cm] (s2) at (6.2, 2.3) {
  $+\;\alpha \cdot c(u) \cdot s_1$
};
\node[resbox] (r2) at (6.2, 0.7) {$dp_{u,\{1\}}$};
\draw[arr] (s2) -- (r2);

\node[lbl] at (10.4, 3.4) {\textbf{(3)} Place $\SF_2$ at $u$};
\node[stepbox, minimum width=3.2cm] (s3) at (10.4, 2.3) {
  $dp_{u,\{1\}} + N_{u,\SF_2}$
};
\node[winbox] (r3) at (10.4, 0.7) {$dp_{u,\{1,2\}}$};
\draw[arr] (s3) -- (r3);

\draw[flowarr] (r1.east) -- (s2.west);
\draw[flowarr] (r2.east) -- (s3.west);

\draw[gray!25, line width=0.4pt] (13.0, -0.1) -- (13.0, 3.7);

\node[lbl] at (15.3, 3.4) {Output};
\node[sem] (F-sf2)  at (15.3, 2.5) {$\SF_2$};
\node[op]  (F-join) at (15.3, 1.6) {$\Join$};
\node[sem] (F-sf1)  at (14.0, 0.7) {$\SF_1$};
\node[tbl] (F-bk)   at (14.0, -0.1) {$L$};
\node[tbl] (F-rv)   at (16.6, 0.7) {$R$};

\draw[treearr] (F-bk)  -- (F-sf1);
\draw[treearr] (F-sf1) -- (F-join);
\draw[treearr] (F-rv)  -- (F-join);
\draw[treearr] (F-join) -- (F-sf2);

\end{tikzpicture}%
}
\caption{Illustrating the DP cost model from \Cref{alg:dp} at join node $u$.
Starting from the input plan (left) where both filters are pushed down, the DP algorithm considers pulling $\SF_2$ up to $u$.
\textbf{(1)}~Distribute $\SF_1$ to $L$, leaving $R$ unfiltered.
\textbf{(2)}~Add relational cost of $u$, scaled by $\alpha$ and reduced by $s_1$.
\textbf{(3)}~Place $\SF_2$ at $u$; the join eliminates unmatched tuples, so the number of distinct $R$-tuples at $u$.
The output plan (right) shows the resulting placement.
The DP algorithm explores all such traces and picks the minimum.}
\Description{Plan tree and three DP steps at the join node showing distribute, add relational cost, and compare placement.}
\label{fig:dp-example}
\end{figure*}
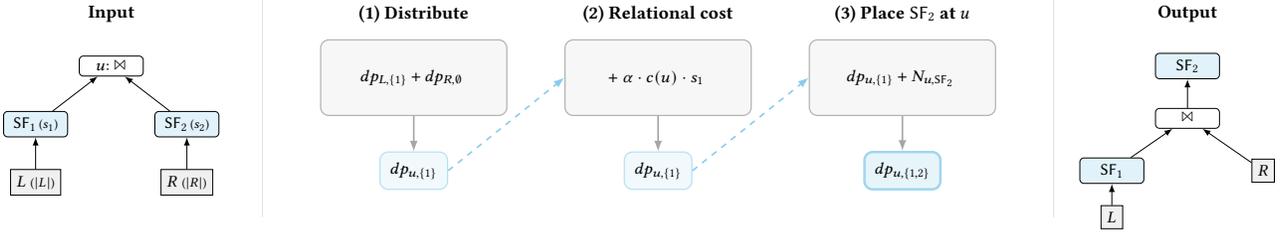

\subsection{DP-based Cost Model}
\label{sec:dp}

We now present a dynamic-programming (DP) algorithm that finds the optimal placement of each semantic filter, balancing LLM and relational cost as defined in Problem~\ref{prob:main}.

\paragraph{Overview.}
The key idea is to traverse the plan tree bottom-up and, at each node $u$, decide which subset of semantic filters to place at that node. This decision reflects a basic tradeoff. Placing a filter below a join reduces the join's input size, but may increase the number of LLM calls because fewer earlier operators have reduced the distinct input count. Placing it above the join can save LLM calls through function caching, but may increase relational cost by leaving the join input unfiltered.
Our DP formulation explores all valid assignments of filters to nodes and selects the one minimizing $C_{\mathrm{LLM}} + \alpha \cdot C_{\mathrm{rel}}$.
To do so, we process each node $u$ by the following three steps:
Step~1 combines the costs from child subtrees;
Step~2 adds the relational cost of operator $u$; and
Step~3 considers placing additional semantic filters at $u$.

\paragraph{Notation.}
Our DP algorithm computes the combined LLM and relational cost at each node for a given set of placed filters.
Let $n$ denote the number of semantic filters after applying simplification in \Cref{sec:simplification}, and let $\SF_1, \ldots, \SF_n$ be these filters.
For each $\SF_i$, let $s_i \in [0, 1]$ denote its selectivity, defined as the fraction of \emph{distinct} tuples it retains under function caching.
Assuming independence between filters, the combined selectivity on a table set $A$ for a set $S \subseteq \{1, \ldots, n\}$ is
\begin{equation}
\text{sel}(A, S) = \prod_{i \in S,\, \text{ref}(\SF_i) \cap A \neq \emptyset} s_i,
\label{eq:selectivity}
\end{equation}

where $\text{ref}(\SF_i)$ denotes the set of base tables that $\SF_i$ references.

To instantiate the cost model, we also need to estimate cardinalities.
For a node $u$, let $\text{tab}(u)$ denote the set of base tables in its subtree.
For the LLM cost of placing $\SF_i$ at node $u$, we need the number of distinct rows from the tables that $\SF_i$ references.
We denote this as $N_{u,\SF_i}$: the estimated number of distinct rows at $u$ projected onto $\text{ref}(\SF_i)$.

Finally, let $c(u)$ denote the per-operator execution cost of $u$ on its original input that only consider the effect of other relational operators but is unfiltered by $\SF$s. 
This can be estimated by the relational optimizer.

\paragraph{DP state and validity.}
Let $dp_{u,S}$ denote the minimum cost of executing the subtree rooted at $u$ when the semantic filters in $S$ have been applied at or below $u$.
Here, not every combination of $u$ and $S$ is valid, and we only consider filters whose original positions are within $u$'s subtree.
A state $dp_{u,S}$ is valid only if, for each $i \in S$, node $u$ lies above $\SF_i$'s original position.
Invalid states are skipped during the computation.
The base case is the leaf node: for a base table $u$, $dp_{u, \emptyset} = 0$.
The DP algorithm accumulates $C_{\mathrm{LLM}} + \alpha \cdot C_{\mathrm{rel}}$ based on its definition, where in step~2, we add the relational cost scaled by $\alpha$, while in step~3, we add the LLM cost directly.

\begin{algorithm}[t]
\caption{DP-based semantic filter placement.}
\label{alg:dp}
\KwIn{Plan tree $\mathcal{T}$, semantic filters $\SF_1, \ldots, \SF_n$, selectivities $s_1, \ldots, s_n$, parameter $\alpha$}
\KwOut{Optimal placement minimizing $C_{\mathrm{LLM}} + \alpha \cdot C_{\mathrm{rel}}$}
Initialize $dp_{u, \emptyset} \leftarrow 0$ for each leaf node $u$; $dp_{u, S} \leftarrow +\infty$ otherwise\nllabel{ln:dp-init}\;
\ForEach{node $u$ in bottom-up order\nllabel{ln:dp-bottomup}}{
    \ForEach{$S \subseteq \{1, \ldots, n\}$ in increasing size\nllabel{ln:dp-subset}}{
        \tcp{Step 1: distribute filters to children}
        \If{$u$ is binary with children $v_1, v_2$\nllabel{ln:dp-binary}}{
            $dp_{u,S} \leftarrow \min\!\big(dp_{u,S},\, \min_{S_1 \subseteq S} dp_{v_1, S_1} + dp_{v_2, S \setminus S_1}\big)$\nllabel{ln:dp-conv}\;
        }
        \ElseIf{$u$ is unary with child $v_1$}{
            $dp_{u,S} \leftarrow dp_{v_1, S}$\;
        }
        \tcp{Step 2: add relational cost at $u$}
        $dp_{u,S} \leftarrow dp_{u,S} + \alpha \cdot c(u) \cdot \text{sel}(\text{tab}(u), S)$\nllabel{ln:dp-rel}\;
        \tcp{Step 3: consider placing each $\SF_i \in S$ at $u$}
        \ForEach{$i \in S$\nllabel{ln:dp-try}}{
            $\ell \leftarrow N_{u,\SF_i} \cdot \text{sel}(\text{ref}(\SF_i), S \setminus \{i\})$\nllabel{ln:dp-llm}\;
            $dp_{u, S} \leftarrow \min(dp_{u, S},\; dp_{u, S \setminus \{i\}} + \ell)$\nllabel{ln:dp-trans}\;
        }
    }
}
\Return $dp_{r, \{1, \ldots, n\}}$ and trace back placement\nllabel{ln:dp-return}\;
\end{algorithm}

\Cref{alg:dp} gives the complete procedure.
The algorithm initializes all states to $+\infty$ except $dp_{u, \emptyset} = 0$ (Line~\ref{ln:dp-init}), then traverses the plan tree bottom-up (Line~\ref{ln:dp-bottomup}), processing subsets in increasing size (Line~\ref{ln:dp-subset}).
We now explain each step in detail.

\paragraph{Step 1: distributing filters to children.}
Each node has at most two children.
For a binary node $u$ with children $v_1, v_2$ (Line~\ref{ln:dp-conv}), we distribute the filters among them:
\begin{equation}
   dp_{u,S} \leftarrow \min\!\Big(dp_{u,S},\; \min_{S_1 \subseteq S} dp_{v_1, S_1} + dp_{v_2, S \setminus S_1}\Big),
\end{equation}
where $S_1$ is assigned to $v_1$ and the remainder $S \setminus S_1$ goes to $v_2$.
Filters whose columns span both children, such as those from $\SJ$ decomposition, cannot go to either child alone, so their child states remain $+\infty$.
Such filters are placed at $u$ via step~3, which runs at smaller subsets before step~1 sees them in $S$.
For a unary node, all filters pass through: $dp_{u, S} = dp_{v_1, S}$.

\paragraph{Step 2: relational cost at $u$.}
After combining the child costs, we add the relational cost of operator $u$ itself.
Since the DP algorithm iterates over subsets $S$ in increasing size order (Line~\ref{ln:dp-subset}), step~3 for subset $S' \subset S$ runs before steps~1--2 for $S$, so all filters in $S$ have valid $dp$ values from earlier iterations.
The filters in $S$ reduce the input, so the adjusted cost is (Line~\ref{ln:dp-rel}):
\begin{equation}
    dp_{u,S} \leftarrow dp_{u,S} + \alpha \cdot c(u) \cdot \text{sel}(\text{tab}(u), S). 
\end{equation}
The factor $\text{sel}(\text{tab}(u), S)$ accounts for the reduction in input size from the semantic filters already placed below $u$.

\paragraph{Step 3: applying a semantic filter at $u$.}
For each non-empty $S$, we enumerate each filter $\SF_i \in S$ and consider placing it at $u$ (Line~\ref{ln:dp-try}).
The LLM cost of evaluating $\SF_i$ at $u$ equals the number of distinct projections on the columns that $\SF_i$ references.
We approximate this by the number of distinct rows from the referenced tables at $u$:
\[
\text{LLM cost of } \SF_i \text{ at } u = N_{u,\SF_i} \cdot \text{sel}(\text{ref}(\SF_i), S \setminus \{i\}).
\]
The transition (Lines~\ref{ln:dp-llm}--\ref{ln:dp-trans}) updates $dp_{u,S}$ by considering placing $\SF_i$ at $u$ for each $i \in S$:
\begin{align*}
dp_{u, S} \leftarrow \min\big(&dp_{u, S},\\
&dp_{u, S \setminus \{i\}} + N_{u,\SF_i} \cdot \text{sel}(\text{ref}(\SF_i), S \setminus \{i\})\big).
\end{align*}


Finally, the optimal cost is $dp_{r, \{1, \ldots, n\}}$ (Line~\ref{ln:dp-return}), meaning all $n$ filters have been placed.
The actual placement is recovered by tracing back the choices at each node.

\Cref{fig:dp-example} provides an example showing the three DP steps above at a join node $u$ with children $L$ and $R$.
The input plan has both $\SF_1$ and $\SF_2$ pushed down.
The figure traces one alternative: keep $\SF_1$ pushed to $L$ (\texttt{Step 1}), add the join cost scaled by $\alpha$ and $s_1$ (\texttt{Step 2}), then place $\SF_2$ at $u$ (\texttt{Step 3}).
The output plan on the right shows this placement.
The DP algorithm explores all such alternatives and selects the minimum-cost plan.

\begin{mythm}[Complexity of \Cref{alg:dp}]
\label{thm:dp-complexity}
The DP algorithm runs in $\mathcal{O}(|V| \cdot n \cdot 2^n + 3^n)$ time, where $|V|$ is the number of nodes in the plan tree and $n$ is the number of semantic filters.
\end{mythm}

\begin{proof}

We first analyze steps 2 and 3. 
At each node, step~2 iterates over $2^n$ subsets in $\mathcal{O}(1)$ each.
In step~3, for each non-empty $S$, we enumerate each $i \in S$ and consider placing $\SF_i$ at $u$.
The total work across all subsets is $\mathcal{O}(n \cdot 2^n)$ per node.
Over all $|V|$ nodes, the total cost of steps~2--3 is $\mathcal{O}(|V| \cdot n \cdot 2^n)$.

For step 1, which is a standard subset convolution problem~\cite{stoian2024dpconv},let $m(u)$ denote the number of semantic filters whose valid range includes binary node $u$. 
Specifically, step~1 at $u$ enumerates all $S_1 \subseteq S$ for each valid $S$, costing $\sum_{k=0}^{m(u)} \binom{m(u)}{k} 2^k = 3^{m(u)}$.
At a binary node $u$ with children $v_1, v_2$, filters are distributed: each filter valid at $u$ goes to at most one child, so $m(v_1) + m(v_2) \leq m(u)$.
Let $T(u)$ be the total cost of all step~1 computations in the subtree of $u$.
We have $T(u) = 3^{m(u)} + T(v_1) + T(v_2) \leq 2 \cdot 3^{m(u)}$.

Therefore, the total time complexity is $\mathcal{O}(|V| \cdot n \cdot 2^n + 3^n)$.
\end{proof}
\section{Implementation}
\label{sec:implementation}

We implement \textsc{PLOP} as an extension to DuckDB~\cite{raasveldt2019duckdb}, modifying its parser, binder, optimizer, and executor with approximately 3,000 lines of C++ code.

\paragraph{Parsing and binding.}
We extend DuckDB’s SQL parser to support semantic operators as native scalar functions: \texttt{SEMANTIC} for $\SF$ and $\SJ$, and typed variants for $\SP$.
During parsing, we split hybrid \texttt{WHERE} clauses into minimal units so that each semantic predicate becomes a separate $\SF$ node that can be independently repositioned.
When \Cref{alg:pullup} pulls an $\SF$ above a projection $\pi$, the columns referenced by $\SF$ must remain available.
We clear and rebuild DuckDB’s projection map after each swap to ensure these columns are not pruned.

\paragraph{Optimizer integration.}
\textsc{PLOP} integrates into DuckDB's optimization pipeline without modifying the native optimizer.
After DuckDB produces a plan with join order and predicate pushdown, \textsc{PLOP} will be executed as a post-processing step.
Semantic filters start at the positions assigned by DuckDB's native optimizer, which typically pushes them down to their lowest feasible positions.
\textsc{PLOP} then considers pulling each filter up from there.
It first applies the simplifications from\Cref{sec:simplification}: $\SP$ pull-up and $\SJ$ decomposition.
The user then selects one of two optimization strategies.
The pull-up algorithm from \Cref{alg:pullup} greedily moves all semantic filters as high as possible.
The DP cost model from \Cref{alg:dp} selectively places each filter to balance LLM and relational cost using DuckDB's cardinality estimates for $N_u$ and $c(u)$.
We evaluate both strategies in \Cref{sec:experiments}.

\paragraph{Cost and cardinality estimation.}
The DP algorithm requires three types of estimates as shown in \Cref{sec:dp}.
The relational cost $c(u)$ at each node is obtained directly from DuckDB's native cardinality estimator.
For semantic filter selectivities used in \Cref{eq:selectivity}, following the approach of join order optimization with minimal statistics~\cite{ebergen2023join}, we set $s_i = 0.2$ for all $\SF$.
We additionally introduce $s_{\Join} = 0.1$ to estimate how much each join reduces the distinct count from one side. 
Cross join (e.g., from $\SJ$ decomposition) is a special case that has selectivity 1 because it does not reduce the count.
For distinct-count estimates $N_{u,\SF_i}$, which represent the number of distinct rows at node $u$ projected onto the tables referred to by $\SF_i$, we traverse the path from the base table to $u$ and multiply the size of the base table by $s_i$ at each semantic filter and by $s_{\Join}$ at each join along the path.
This provides a simple, statistics-free estimate that is consistent with the selectivity model used in the DP recurrence.
Fine-grained estimation via sampling or learned models is complementary and can replace these fixed defaults.
We evaluate sensitivity to these estimates in \Cref{sec:exp-ablations}.

\paragraph{Function caching.}
We implement function caching using a single concurrent hash table shared across all operators in the execution pipeline, with bucket-level locking for parallel reads and writes during vectorized execution.
The cache is keyed on the rendered prompt string, which includes both the predicate $\phi$ and the input tuple values, so different predicates never share cache entries.
On a cache hit, the LLM call is skipped entirely.
The cache is scoped per query execution and cleared between queries.

However, function caching is not free: cache lookups add relational overhead when a semantic filter is pulled above a join, since every output row triggers a cache probe.
The cost model accounts for this by including lookup cost in the relational cost at the join node, estimated by the join output size times the number of semantic filters at ancestors of that join.

\paragraph{Execution.}
Lastly, query execution follows DuckDB's push-based vectorized pipeline.
Semantic operators are executed as scalar functions within the pipeline, with function caching intercepting redundant LLM calls.
In \textsc{PLOP}, no changes to DuckDB's execution engine are required beyond registering the semantic functions and the function cache lookup.

\section{Evaluation}
\label{sec:experiments}

We evaluate \textsc{PLOP} on two benchmarks and address the following research questions:
\begin{itemize}[leftmargin=*, itemsep=0.5ex]
    \item \textbf{RQ1:} Does \textsc{PLOP} improve latency and reduce LLM cost while maintaining accuracy compared to existing systems?
    \item \textbf{RQ2:} Does accuracy hold against human-annotated ground truth, not just the baseline output?
    \item \textbf{RQ3:} When does pull-up suffice, and when does using \textsc{PLOP}'s cost model provide additional benefit?
    \item \textbf{RQ4:} How sensitive is the cost model to its parameters ($\alpha$, selectivity estimates), and what is the optimizer overhead?
\end{itemize}

\subsection{Experimental Setup}
\label{sec:exp-setup}

\paragraph{Environment.}
All experiments run on an AWS \texttt{r7i.8xlarge} instance (32 vCPUs, 256 GB RAM).
We use GPT-5-mini as the LLM backend via the OpenAI API.
The cost model parameter is set to $\alpha = 10^{-7}$, reflecting the large gap between per-row relational processing cost and per-call LLM inference cost.

\paragraph{Hybrid query benchmark.}
Existing semantic query benchmarks such as SemBench~\cite{lao2026sembenchbenchmarksemanticquery} focus on simple queries with few relational and semantic operators (most $\leq 3$), which offer limited optimization opportunity for plan-level placement.
To stress-test \textsc{PLOP} on complex hybrid plans, we construct a benchmark of 30 semantic SQL queries spanning four schemas based on DataAgentBench~\cite{ma2026can} and TPC-H~\cite{tpch2024}.
The original queries from \cite{ma2026can} are written in natural language or basic SQL for a database agent benchmark.
We adapt them into hybrid SQL by translating natural language intent into semantic operators ($\SP$, $\SF$, $\SJ$) composed with relational operators.
For TPC-H, we augment standard analytical queries with semantic predicates over text-rich columns.
\Cref{tab:schema-sizes} summarizes the schema sizes.
\Cref{fig:query-char} shows the operator composition of each query in the hybrid query benchmark.
The queries cover all three semantic operator types:
Q1--Q3 use $\SP$ to summarize or score text fields.
Q4--Q30 use $\SF$s.
Q16--Q17, Q25, and Q27--Q30 additionally use $\SJ$ to match tuples across tables.
Query complexity ranges from two table queries with one semantic operator to multi-way joins with up to 9 relational joins and 4 semantic filters. 

\begin{table}[t]
\centering
\caption{Hybrid query benchmark: schema statistics (Avg Rows rounded to nearest integer).}
\label{tab:schema-sizes}
\renewcommand{\arraystretch}{1.15}
\begin{tabular}{l r r l}
\toprule
\textbf{Schema} & \textbf{Tables} & \textbf{Avg Rows} & \textbf{Source} \\
\midrule
BookReview & 3 & 682 & \cite{ma2026can} \\
Yelp & 3 & 1{,}594 & \cite{ma2026can} \\
GoogleLocal & 2 & 1{,}040 & \cite{ma2026can} \\
TPC-H (SF=0.005) & 8 & 10{,}851 & \cite{tpch2024} \\
\bottomrule
\end{tabular}
\end{table}

\definecolor{colSP}{HTML}{E07B54}     
\definecolor{colSF}{HTML}{F2A63B}     
\definecolor{colSJc}{HTML}{C75D8A}    
\definecolor{colSig}{HTML}{6BA3C7}    
\definecolor{colJoin}{HTML}{3D7A99}   

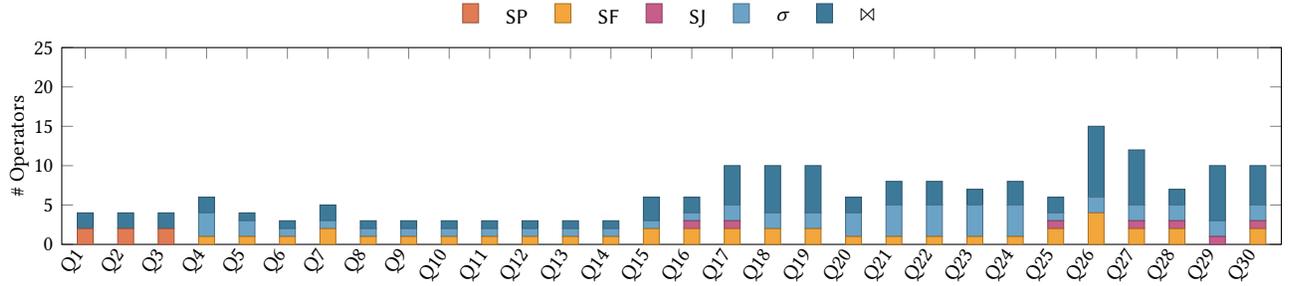
\begin{figure*}[t]
\centering
\begin{tikzpicture}
\begin{axis}[
  ybar stacked,
  width=\textwidth,
  height=4.2cm,
  enlarge x limits=0.02,
  bar width=6pt,
  symbolic x coords={Q1,Q2,Q3,Q4,Q5,Q6,Q7,Q8,Q9,Q10,Q11,Q12,Q13,Q14,Q15,Q16,Q17,Q18,Q19,Q20,Q21,Q22,Q23,Q24,Q25,Q26,Q27,Q28,Q29,Q30},
  xtick=data,
  xticklabel style={font=\small, rotate=55, anchor=east},
  ymin=0, ymax=25,
  ytick={0,5,10,15,20,25,30},
  ylabel={\# Operators},
  ylabel style={font=\small, at={(axis description cs:-0.02,0.5)}},
  yticklabel style={font=\small},
  ymajorgrids=false,
  legend style={
    font=\small,
    at={(0.5,1.05)},
    anchor=south,
    legend columns=5,
    column sep=8pt,
    draw=none,
  },
]

\addplot[fill=colSP, draw=colSP!70!black] coordinates{
  (Q1,2)(Q2,2)(Q3,2)(Q4,0)(Q5,0)(Q6,0)(Q7,0)(Q8,0)(Q9,0)(Q10,0)
  (Q11,0)(Q12,0)(Q13,0)(Q14,0)(Q15,0)(Q16,0)(Q17,0)(Q18,0)(Q19,0)(Q20,0)
  (Q21,0)(Q22,0)(Q23,0)(Q24,0)(Q25,0)(Q26,0)(Q27,0)(Q28,0)(Q29,0)(Q30,0)};

\addplot[fill=colSF, draw=colSF!70!black] coordinates{
  (Q1,0)(Q2,0)(Q3,0)(Q4,1)(Q5,1)(Q6,1)(Q7,2)(Q8,1)(Q9,1)(Q10,1)
  (Q11,1)(Q12,1)(Q13,1)(Q14,1)(Q15,2)(Q16,2)(Q17,2)(Q18,2)(Q19,2)(Q20,1)
  (Q21,1)(Q22,1)(Q23,1)(Q24,1)(Q25,2)(Q26,4)(Q27,2)(Q28,2)(Q29,0)(Q30,2)};

\addplot[fill=colSJc, draw=colSJc!70!black] coordinates{
  (Q1,0)(Q2,0)(Q3,0)(Q4,0)(Q5,0)(Q6,0)(Q7,0)(Q8,0)(Q9,0)(Q10,0)
  (Q11,0)(Q12,0)(Q13,0)(Q14,0)(Q15,0)(Q16,1)(Q17,1)(Q18,0)(Q19,0)(Q20,0)
  (Q21,0)(Q22,0)(Q23,0)(Q24,0)(Q25,1)(Q26,0)(Q27,1)(Q28,1)(Q29,1)(Q30,1)};

\addplot[fill=colSig, draw=colSig!70!black] coordinates{
  (Q1,0)(Q2,0)(Q3,0)(Q4,3)(Q5,2)(Q6,1)(Q7,1)(Q8,1)(Q9,1)(Q10,1)
  (Q11,1)(Q12,1)(Q13,1)(Q14,1)(Q15,1)(Q16,1)(Q17,2)(Q18,2)(Q19,2)(Q20,3)
  (Q21,4)(Q22,4)(Q23,4)(Q24,4)(Q25,1)(Q26,2)(Q27,2)(Q28,2)(Q29,2)(Q30,2)};

\addplot[fill=colJoin, draw=colJoin!70!black] coordinates{
  (Q1,2)(Q2,2)(Q3,2)(Q4,2)(Q5,1)(Q6,1)(Q7,2)(Q8,1)(Q9,1)(Q10,1)
  (Q11,1)(Q12,1)(Q13,1)(Q14,1)(Q15,3)(Q16,2)(Q17,5)(Q18,6)(Q19,6)(Q20,2)
  (Q21,3)(Q22,3)(Q23,2)(Q24,3)(Q25,2)(Q26,9)(Q27,7)(Q28,2)(Q29,7)(Q30,5)};

\legend{$\SP$, $\SF$, $\SJ$, $\sigma$, $\Join$}

\end{axis}
\end{tikzpicture}
\caption{Hybrid query benchmark: operator composition per query. Each stacked bar shows the count of $\SP$, $\SF$, $\SJ$, $\sigma$, and $\Join$ operators. Q17--Q19 and Q26--Q30 are the most complex, with up to 9 joins and 4 semantic filters.}
\Description{Stacked bar chart showing operator counts per query across the 30-query hybrid benchmark.}
\label{fig:query-char}
\end{figure*}

\paragraph{SemBench E-Commerce.}
In our hybrid queries, we compute the quality of each method by referring to the results of the initial DuckDB UDF execution as ground truth, which does not have semantic-related query rewriting that will introduce additional errors.
All errors come from LLM's non-determinism across invocations.
To further validate the accuracy against ground truth, we additionally evaluate on the E-Commerce subset of SemBench~\cite{lao2026sembenchbenchmarksemanticquery}, which has human-annotated labels.
This benchmark contains 14 queries over an e-commerce product catalog.
These queries are simpler than ours and typically involve fewer than 3 relational operators, so the main question is whether \textsc{PLOP} maintains high accuracy while still reducing cost and latency.

\paragraph{Baselines.}
We compare against the following systems.
All use GPT-5-mini as the sole LLM backend.
\begin{itemize}[leftmargin=*, itemsep=0.5ex]
    \item \textit{DuckDB + Cache}: semantic operators executed as DuckDB UDFs with function caching enabled but no placement optimization. This isolates the gains from \textsc{PLOP}'s placement decisions.
    \item \textit{ThalamusDB}~\cite{jo2024thalamusdb} (v0.1.15): approximate query processing for multi-modal data. Does not support $\SP$, so queries with semantic projections are excluded.
    \item \textit{FlockMTL}~\cite{dorbani2025beyond} (v0.7.0): LLM operator integration in DuckDB with its own prompt and batching optimizations.
    \item \textit{Palimpzest}~\cite{liu2025palimpzest} \textit{and Abacus}~\cite{russo2025abacuscostbasedoptimizersemantic} (v1.4.0): declarative LLM analytics with cost-based planning. Palimpzest optimizes individual operator implementations, while Abacus adds cost-based physical plan selection. We evaluate both.
    \item \textit{Lotus}~\cite{patel2025semantic} (v1.1.3): model cascading that routes easy cases to a cheaper proxy model and harder cases to GPT-5-mini. Following the SemBench configuration~\cite{lao2026sembenchbenchmarksemanticquery}, Lotus uses e5-base-v2 for text embeddings and CLIP-ViT-B-32 for image embeddings. Because its multi-model setup makes cost and latency not directly comparable, we include Lotus on SemBench only for accuracy comparison.
\end{itemize}

\paragraph{Metrics.}
We report three metrics.
\emph{Speedup} is the geometric mean of the per-query latency improvement relative to DuckDB + Cache.
\emph{Cost Reduction} is the geometric mean of the per-query LLM cost reduction relative to DuckDB + Cache.
Here, LLM cost refers to the total OpenAI API cost in dollars incurred during a query execution, including both input and output tokens.
We also report the quality of the results. On the 30-query hybrid benchmark, we define quality as the arithmetic mean of per-query F1 scores, using the output of DuckDB + Cache as the reference.
On SemBench, we report the benchmark-defined quality metric against the provided human-annotated ground truth.
We run each query with a timeout of 3,000 seconds.

\subsection{Hybrid Query Benchmark Results}
\label{sec:exp-e2e}

\begin{table}[t]
\centering
\caption{Overall performance on the 30-query hybrid benchmark. Speedup and cost reduction are geometric means vs.\ DuckDB + Cache baseline. F1 is the arithmetic mean.}
\label{tab:summary}
\small
\renewcommand{\arraystretch}{1.15}
\resizebox{0.9\linewidth}{!}{%
\begin{tabular}{l r r r}
\toprule
\textbf{Method} & \textbf{Speedup} & \textbf{Cost Red.} & \textbf{Avg F1} \\
\midrule
\rowcolor{gray!12} Baseline (DuckDB + Cache) & 1.00$\times$ & 1.00$\times$ & 1.000 \\[1pt]
\rowcolor{colCost!12} \textsc{PLOP}-Cost & 1.50$\times$ & 4.18$\times$ & 0.848 \\
\rowcolor{colCost!12} \textsc{PLOP}-Pullup & 1.16$\times$ & 4.29$\times$ & 0.849 \\[1pt]
ThalamusDB$^\dagger$ & 1.50$\times$ & 10.10$\times$ & 0.515 \\
FlockMTL & 0.64$\times$ & 7.31$\times$ & 0.343 \\
Palimpzest & 1.70$\times$ & 2.15$\times$ & 0.533 \\
Abacus & 0.15$\times$ & 0.43$\times$ & 0.561 \\
\bottomrule
\end{tabular}%
}
\par\vspace{2pt}
{\footnotesize $^\dagger$\, Q1--Q3 are excluded from ThalamusDB as $\SP$ is not supported.}
\end{table}


\pgfplotsset{
  failmark/.style={
    only marks,
    mark=x,
    mark size=3.5pt,
    line width=1.4pt,
  },
  nonebar/.style={bar shift=-18pt},
  costbar/.style={bar shift=-12pt},
  pullbar/.style={bar shift=-6pt},
  thalbar/.style={bar shift=0pt},
  flockbar/.style={bar shift=6pt},
  palmbar/.style={bar shift=12pt},
  ababar/.style={bar shift=18pt},
  compactlegend/.style={
    legend columns=7,
    legend cell align=left,
    legend style={
      font=\Large,
      at={(0.5,1.08)},
      anchor=south,
      draw=none,
      fill=none,
      column sep=12pt,
      /tikz/every even column/.append style={column sep=10pt}
    }
  }
}

\definecolor{colAba}{RGB}{0,122,204}

\begin{figure*}[t!]
\centering

\subfigure[Latency Q1--Q15]{%
\resizebox{0.95\linewidth}{!}{%
\begin{tikzpicture}[trim axis left, trim axis right]
\begin{axis}[FIGALL,compactlegend,
  enlarge x limits=0.04,
  symbolic x coords={Q1,Q2,Q3,Q4,Q5,Q6,Q7,Q8,Q9,Q10,Q11,Q12,Q13,Q14,Q15},
  xtick={Q1,Q2,Q3,Q4,Q5,Q6,Q7,Q8,Q9,Q10,Q11,Q12,Q13,Q14,Q15},
  ymode=log, log origin=infty,
  ymin=1, ymax=5000,
  ytick={1,10,100,1000},
  yticklabels={$10^0$,$10^1$,$10^2$,$10^3$},
  ylabel={Time (s)},
  extra x ticks={Q2,Q4,Q6,Q8,Q10,Q12,Q14},
  extra x tick style={grid=major, grid style={draw=gray!15, line width=1.7cm}},
  extra x tick labels={},
]
\addplot[nonebar,fill=colNone,draw=colNone!70!black] coordinates{
  (Q1,12.56)(Q2,7.95)(Q3,9.78)
  (Q4,19.16)(Q5,5.15)(Q6,14.34)
  (Q7,12.79)(Q8,11.52)(Q9,11.03)
  (Q10,17.09)(Q11,12.11)(Q12,9.98)
  (Q13,12.84)(Q14,9.19)(Q15,75.46)};
\addplot[costbar,fill=colCost,draw=colCost!70!black] coordinates{
  (Q1,11.81)(Q2,17.15)(Q3,8.29)
  (Q4,14.64)(Q5,5.33)(Q6,7.21)
  (Q7,10.55)(Q8,2.92)(Q9,2.71)
  (Q10,10.29)(Q11,12.22)(Q12,6.32)
  (Q13,8.35)(Q14,4.11)(Q15,25.09)};
\addplot[pullbar,fill=colPull,draw=colPull!70!black] coordinates{
  (Q1,12.34)(Q2,7.78)(Q3,8.66)
  (Q4,14.82)(Q5,4.53)(Q6,3.00)
  (Q7,12.83)(Q8,2.08)(Q9,2.42)
  (Q10,8.73)(Q11,10.97)(Q12,9.47)
  (Q13,6.24)(Q14,5.22)(Q15,65.13)};
\addplot[thalbar,fill=colThal,draw=colThal!70!black] coordinates{
  (Q1,0)(Q2,0)(Q3,0)(Q4,0)(Q5,1.42)(Q6,14.92)
  (Q7,0)(Q8,2.07)(Q9,16.83)(Q10,6.57)(Q11,6.10)(Q12,4.93)
  (Q13,0)(Q14,6.62)(Q15,0)};
\addplot[flockbar,fill=colFlock,draw=colFlock!70!black] coordinates{
  (Q1,61.82)(Q2,0)(Q3,37.42)(Q4,0)(Q5,6.12)(Q6,0)
  (Q7,0)(Q8,0)(Q9,77.06)(Q10,8.91)(Q11,0)(Q12,9.94)
  (Q13,0)(Q14,0)(Q15,0)};
\addplot[palmbar,fill=colPalm,draw=colPalm!70!black] coordinates{
  (Q1,0)(Q2,0)(Q3,0)(Q4,0)(Q5,2.51)(Q6,6.53)
  (Q7,0)(Q8,6.81)(Q9,0)(Q10,4.42)(Q11,3.41)(Q12,2.42)
  (Q13,4.38)(Q14,2.64)(Q15,0)};
\addplot[ababar,fill=colAba,draw=colAba!70!black] coordinates{
  (Q1,0)(Q2,0)(Q3,0)(Q4,0)(Q5,9.07)(Q6,313.66)
  (Q7,0)(Q8,320.82)(Q9,344.72)(Q10,58.19)(Q11,336.92)(Q12,58.42)
  (Q13,310.43)(Q14,68.85)(Q15,0)};
\addplot[failmark, draw=colThal, xshift=0pt] coordinates{
  (Q1,1.3)(Q2,1.3)(Q3,1.3)(Q4,1.3)(Q7,1.3)(Q13,1.3)(Q15,1.3)};
\addplot[failmark, draw=colFlock, xshift=6pt] coordinates{
  (Q2,1.3)(Q4,1.3)(Q6,1.3)(Q7,1.3)
  (Q8,1.3)(Q11,1.3)(Q13,1.3)(Q14,1.3)(Q15,1.3)};
\addplot[failmark, draw=colPalm, xshift=12pt] coordinates{
  (Q1,1.3)(Q2,1.3)(Q3,1.3)(Q4,1.3)(Q7,1.3)(Q9,1.3)(Q15,1.3)};
\addplot[failmark, draw=colAba, xshift=18pt] coordinates{
  (Q1,1.3)(Q2,1.3)(Q3,1.3)(Q4,1.3)(Q7,1.3)(Q15,1.3)};
\addplot[failmark, draw=colCost, xshift=-12pt] coordinates{
};
\addplot[failmark, draw=colPull, xshift=-6pt] coordinates{
};
\addlegendimage{legend image code/.code={
  \fill[colNone]   (0,-.1) rectangle (.4,.35);
  \fill[colNone]   (.4,-.1) rectangle (.8,.35);
  \filldraw[pattern=north east lines,pattern color=black,draw=colNone!70!black] (.4,-.1) rectangle (.8,.35);
  \draw[colNone!70!black] (0,-.1) rectangle (.4,.35);}}
\addlegendentry{Baseline (DuckDB + Cache)}
\addlegendimage{legend image code/.code={
  \fill[colCost]   (0,-.1) rectangle (.4,.35);
  \fill[colCost]   (.4,-.1) rectangle (.8,.35);
  \filldraw[pattern=north east lines,pattern color=black,draw=colCost!70!black] (.4,-.1) rectangle (.8,.35);
  \draw[colCost!70!black] (0,-.1) rectangle (.4,.35);}}
\addlegendentry{\textsc{PLOP}-Cost (Ours)}
\addlegendimage{legend image code/.code={
  \fill[colPull]   (0,-.1) rectangle (.4,.35);
  \fill[colPull]   (.4,-.1) rectangle (.8,.35);
  \filldraw[pattern=north east lines,pattern color=black,draw=colPull!70!black] (.4,-.1) rectangle (.8,.35);
  \draw[colPull!70!black] (0,-.1) rectangle (.4,.35);}}
\addlegendentry{\textsc{PLOP}-Pullup (Ours)}
\addlegendimage{legend image code/.code={
  \fill[colThal]   (0,-.1) rectangle (.4,.35);
  \fill[colThal]   (.4,-.1) rectangle (.8,.35);
  \filldraw[pattern=north east lines,pattern color=black,draw=colThal!70!black] (.4,-.1) rectangle (.8,.35);
  \draw[colThal!70!black] (0,-.1) rectangle (.4,.35);}}
\addlegendentry{ThalamusDB}
\addlegendimage{legend image code/.code={
  \fill[colFlock]  (0,-.1) rectangle (.4,.35);
  \fill[colFlock]  (.4,-.1) rectangle (.8,.35);
  \filldraw[pattern=north east lines,pattern color=black,draw=colFlock!70!black] (.4,-.1) rectangle (.8,.35);
  \draw[colFlock!70!black] (0,-.1) rectangle (.4,.35);}}
\addlegendentry{FlockMTL}
\addlegendimage{legend image code/.code={
  \fill[colPalm]   (0,-.1) rectangle (.4,.35);
  \fill[colPalm]   (.4,-.1) rectangle (.8,.35);
  \filldraw[pattern=north east lines,pattern color=black,draw=colPalm!70!black] (.4,-.1) rectangle (.8,.35);
  \draw[colPalm!70!black] (0,-.1) rectangle (.4,.35);}}
\addlegendentry{Palimpzest}
\addlegendimage{legend image code/.code={
  \fill[colAba]   (0,-.1) rectangle (.4,.35);
  \fill[colAba]   (.4,-.1) rectangle (.8,.35);
  \filldraw[pattern=north east lines,pattern color=black,draw=colAba!70!black] (.4,-.1) rectangle (.8,.35);
  \draw[colAba!70!black] (0,-.1) rectangle (.4,.35);}}
\addlegendentry{Abacus}
\end{axis}
\end{tikzpicture}%
}%
}
\vspace{-3ex}

\subfigure[Cost Q1--Q15]{%
\resizebox{0.95\linewidth}{!}{%
\begin{tikzpicture}[trim axis left, trim axis right]
\begin{axis}[FIGALLbot,
  enlarge x limits=0.04,
  symbolic x coords={Q1,Q2,Q3,Q4,Q5,Q6,Q7,Q8,Q9,Q10,Q11,Q12,Q13,Q14,Q15},
  xtick={Q1,Q2,Q3,Q4,Q5,Q6,Q7,Q8,Q9,Q10,Q11,Q12,Q13,Q14,Q15},
  ymode=log, log origin=infty,
  ymin=0.00008, ymax=10,
  ytick={0.0001,0.001,0.01,0.1,1,10},
  yticklabels={$10^{-4}$,$10^{-3}$,$10^{-2}$,$10^{-1}$,$10^{0}$,$10^{1}$},
  ylabel={Cost (\$)},
  extra x ticks={Q2,Q4,Q6,Q8,Q10,Q12,Q14},
  extra x tick style={grid=major, grid style={draw=gray!15, line width=1.7cm}},
  extra x tick labels={},
]
\addplot[nonebar,fill=colNone,draw=colNone!70!black,postaction={pattern=north east lines,pattern color=black}] coordinates{
  (Q1,0.050356)(Q2,0.036261)(Q3,0.022509)
  (Q4,0.140485)(Q5,0.000645)(Q6,0.088432)
  (Q7,0.095528)(Q8,0.095702)(Q9,0.088049)
  (Q10,0.015840)(Q11,0.044204)(Q12,0.013978)
  (Q13,0.040655)(Q14,0.030009)(Q15,0.782416)};
\addplot[costbar,fill=colCost,draw=colCost!70!black,postaction={pattern=north east lines,pattern color=black}] coordinates{
  (Q1,0.048858)(Q2,0.035716)(Q3,0.023935)
  (Q4,0.103578)(Q5,0.001025)(Q6,0.004653)
  (Q7,0.012006)(Q8,0.000644)(Q9,0.000565)
  (Q10,0.012625)(Q11,0.007755)(Q12,0.004677)
  (Q13,0.002848)(Q14,0.000906)(Q15,0.067283)};
\addplot[pullbar,fill=colPull,draw=colPull!70!black,postaction={pattern=north east lines,pattern color=black}] coordinates{
  (Q1,0.050970)(Q2,0.035713)(Q3,0.022779)
  (Q4,0.102558)(Q5,0.000585)(Q6,0.004671)
  (Q7,0.011838)(Q8,0.000612)(Q9,0.000527)
  (Q10,0.010743)(Q11,0.008426)(Q12,0.004873)
  (Q13,0.002572)(Q14,0.001173)(Q15,0.066313)};
\addplot[thalbar,fill=colThal,draw=colThal!70!black,postaction={pattern=north east lines,pattern color=black}] coordinates{
  (Q1,0)(Q2,0)(Q3,0)(Q4,0)(Q5,0.000074)(Q6,0.036188)
  (Q7,0)(Q8,0.003517)(Q9,0.044451)(Q10,0.003979)(Q11,0.003979)(Q12,0.003047)
  (Q13,0)(Q14,0.004038)(Q15,0)};
\addplot[flockbar,fill=colFlock,draw=colFlock!70!black,postaction={pattern=north east lines,pattern color=black}] coordinates{
  (Q1,0.021500)(Q2,0)(Q3,0.007900)(Q4,0)(Q5,0.001000)(Q6,0)
  (Q7,0)(Q8,0)(Q9,0.049900)(Q10,0.002200)(Q11,0)(Q12,0.002200)
  (Q13,0)(Q14,0)(Q15,0)};
\addplot[palmbar,fill=colPalm,draw=colPalm!70!black,postaction={pattern=north east lines,pattern color=black}] coordinates{
  (Q1,0)(Q2,0)(Q3,0)(Q4,0)(Q5,0.000688)(Q6,0.089826)
  (Q7,0)(Q8,0.086376)(Q9,0)(Q10,0.006342)(Q11,0.020079)(Q12,0.006151)
  (Q13,0.018877)(Q14,0.014106)(Q15,0)};
\addplot[ababar,fill=colAba,draw=colAba!70!black,postaction={pattern=north east lines,pattern color=black}] coordinates{
  (Q1,0)(Q2,0)(Q3,0)(Q4,0)(Q5,0.001337)(Q6,0.235676)
  (Q7,0)(Q8,0.227808)(Q9,0.229677)(Q10,0.024677)(Q11,0.118774)(Q12,0.025597)
  (Q13,0.121736)(Q14,0.039488)(Q15,0)};
\addplot[failmark, draw=colThal, xshift=0pt] coordinates{
  (Q1,0.00011)(Q2,0.00011)(Q3,0.00011)(Q4,0.00011)(Q7,0.00011)(Q13,0.00011)(Q15,0.00011)};
\addplot[failmark, draw=colFlock, xshift=6pt] coordinates{
  (Q2,0.00011)(Q4,0.00011)(Q6,0.00011)(Q7,0.00011)
  (Q8,0.00011)(Q11,0.00011)(Q13,0.00011)(Q14,0.00011)(Q15,0.00011)};
\addplot[failmark, draw=colPalm, xshift=12pt] coordinates{
  (Q1,0.00011)(Q2,0.00011)(Q3,0.00011)(Q4,0.00011)(Q7,0.00011)(Q9,0.00011)(Q15,0.00011)};
\addplot[failmark, draw=colAba, xshift=18pt] coordinates{
  (Q1,0.00011)(Q2,0.00011)(Q3,0.00011)(Q4,0.00011)(Q7,0.00011)(Q15,0.00011)};
\addplot[failmark, draw=colCost, xshift=-12pt] coordinates{
};
\addplot[failmark, draw=colPull, xshift=-6pt] coordinates{
};
\end{axis}
\end{tikzpicture}%
}%
}
\vspace{-3ex}

\subfigure[Latency Q16--Q30]{%
\resizebox{0.95\linewidth}{!}{%
\begin{tikzpicture}[trim axis left, trim axis right]
\begin{axis}[FIGALLbot,
  enlarge x limits=0.04,
  symbolic x coords={Q16,Q17,Q18,Q19,Q20,Q21,Q22,Q23,Q24,Q25,Q26,Q27,Q28,Q29,Q30},
  xtick={Q16,Q17,Q18,Q19,Q20,Q21,Q22,Q23,Q24,Q25,Q26,Q27,Q28,Q29,Q30},
  ymode=log, log origin=infty,
  ymin=1, ymax=5000,
  ytick={1,10,100,1000},
  yticklabels={$10^0$,$10^1$,$10^2$,$10^3$},
  ylabel={Time (s)},
  extra x ticks={Q17,Q19,Q21,Q23,Q25,Q27,Q29},
  extra x tick style={grid=major, grid style={draw=gray!15, line width=1.7cm}},
  extra x tick labels={},
]
\addplot[nonebar,fill=colNone,draw=colNone!70!black] coordinates{
  (Q16,81.26)(Q17,208.66)(Q18,26.24)
  (Q19,626.54)(Q20,4.72)(Q21,8.17)
  (Q22,12.17)(Q23,7.89)(Q24,11.80)
  (Q25,76.07)(Q26,255.06)(Q27,115.63)
  (Q28,83.29)(Q29,11.33)(Q30,33.40)};
\addplot[costbar,fill=colCost,draw=colCost!70!black] coordinates{
  (Q16,0)(Q17,181.49)(Q18,30.09)
  (Q19,442.97)(Q20,4.54)(Q21,7.76)
  (Q22,5.60)(Q23,11.87)(Q24,9.25)
  (Q25,0)(Q26,75.80)(Q27,78.99)
  (Q28,30.97)(Q29,11.26)(Q30,28.20)};
\addplot[pullbar,fill=colPull,draw=colPull!70!black] coordinates{
  (Q16,102.85)(Q17,1405.61)(Q18,89.23)
  (Q19,0)(Q20,2.85)(Q21,7.40)
  (Q22,7.01)(Q23,10.14)(Q24,9.25)
  (Q25,111.95)(Q26,75.22)(Q27,93.98)
  (Q28,28.90)(Q29,52.66)(Q30,27.64)};
\addplot[thalbar,fill=colThal,draw=colThal!70!black] coordinates{
  (Q16,0)(Q17,0)(Q18,226.56)(Q19,0)(Q20,0)(Q21,8.59)
  (Q22,0)(Q23,9.56)(Q24,9.38)(Q25,0)(Q26,0)(Q27,0)
  (Q28,0)(Q29,1.12)(Q30,0)};
\addplot[flockbar,fill=colFlock,draw=colFlock!70!black] coordinates{
  (Q16,0)(Q17,0)(Q18,0)(Q19,0)(Q20,3.26)(Q21,0)
  (Q22,25.51)(Q23,0)(Q24,0)(Q25,0)(Q26,0)(Q27,0)
  (Q28,0)(Q29,2.72)(Q30,20.09)};
\addplot[palmbar,fill=colPalm,draw=colPalm!70!black] coordinates{
  (Q16,53.26)(Q17,0)(Q18,29.68)(Q19,0)(Q20,5.25)(Q21,16.00)
  (Q22,7.14)(Q23,6.92)(Q24,3.65)(Q25,0)(Q26,63.73)(Q27,282.47)
  (Q28,0)(Q29,0)(Q30,0)};
\addplot[ababar,fill=colAba,draw=colAba!70!black] coordinates{
  (Q16,139.99)(Q17,0)(Q18,0)(Q19,0)(Q20,827.26)(Q21,656.21)
  (Q22,691.01)(Q23,0)(Q24,1063.64)(Q25,0)(Q26,188.38)(Q27,329.53)
  (Q28,0)(Q29,0)(Q30,0)};
\addplot[failmark, draw=colThal, xshift=0pt] coordinates{
  (Q16,1.3)(Q17,1.3)(Q19,1.3)(Q20,1.3)
  (Q22,1.3)(Q25,1.3)(Q26,1.3)(Q27,1.3)(Q28,1.3)(Q30,1.3)};
\addplot[failmark, draw=colFlock, xshift=6pt] coordinates{
  (Q16,1.3)(Q17,1.3)(Q18,1.3)(Q19,1.3)(Q21,1.3)(Q23,1.3)
  (Q24,1.3)(Q25,1.3)(Q26,1.3)(Q27,1.3)(Q28,1.3)};
\addplot[failmark, draw=colPalm, xshift=12pt] coordinates{
  (Q17,1.3)(Q19,1.3)
  (Q25,1.3)(Q28,1.3)(Q29,1.3)(Q30,1.3)};
\addplot[failmark, draw=colAba, xshift=18pt] coordinates{
  (Q17,1.3)(Q18,1.3)(Q19,1.3)(Q23,1.3)
  (Q25,1.3)(Q28,1.3)(Q29,1.3)(Q30,1.3)};
\addplot[failmark, draw=colCost, xshift=-12pt] coordinates{
  (Q16,1.3)(Q25,1.3)};
\addplot[failmark, draw=colPull, xshift=-6pt] coordinates{
  (Q19,1.3)};
\end{axis}
\end{tikzpicture}%
}%
}

\vspace{-3ex}

\subfigure[Cost Q16--Q30]{%
\resizebox{0.95\linewidth}{!}{%
\begin{tikzpicture}[trim axis left, trim axis right]
\begin{axis}[FIGALLbot,
  enlarge x limits=0.04,
  symbolic x coords={Q16,Q17,Q18,Q19,Q20,Q21,Q22,Q23,Q24,Q25,Q26,Q27,Q28,Q29,Q30},
  xtick={Q16,Q17,Q18,Q19,Q20,Q21,Q22,Q23,Q24,Q25,Q26,Q27,Q28,Q29,Q30},
  ymode=log, log origin=infty,
  ymin=0.00008, ymax=10,
  ytick={0.0001,0.001,0.01,0.1,1,10},
  yticklabels={$10^{-4}$,$10^{-3}$,$10^{-2}$,$10^{-1}$,$10^{0}$,$10^{1}$},
  ylabel={Cost (\$)},
  extra x ticks={Q17,Q19,Q21,Q23,Q25,Q27,Q29},
  extra x tick style={grid=major, grid style={draw=gray!15, line width=1.7cm}},
  extra x tick labels={},
]
\addplot[nonebar,fill=colNone,draw=colNone!70!black,postaction={pattern=north east lines,pattern color=black}] coordinates{
  (Q16,0.796040)(Q17,2.271902)(Q18,0.051324)
  (Q19,7.168478)(Q20,0.002524)(Q21,0.031397)
  (Q22,0.036050)(Q23,0.040919)(Q24,0.049782)
  (Q25,0.818454)(Q26,1.586768)(Q27,0.777081)
  (Q28,0.291956)(Q29,0.006276)(Q30,0.043414)};
\addplot[costbar,fill=colCost,draw=colCost!70!black,postaction={pattern=north east lines,pattern color=black}] coordinates{
  (Q16,0)(Q17,1.304349)(Q18,0.046342)
  (Q19,5.143268)(Q20,0.001203)(Q21,0.030118)
  (Q22,0.004493)(Q23,0.043277)(Q24,0.033461)
  (Q25,0)(Q26,0.259580)(Q27,0.126442)
  (Q28,0.028411)(Q29,0.005589)(Q30,0.009519)};
\addplot[pullbar,fill=colPull,draw=colPull!70!black,postaction={pattern=north east lines,pattern color=black}] coordinates{
  (Q16,0.087848)(Q17,1.313800)(Q18,0.048011)
  (Q19,0)(Q20,0.001361)(Q21,0.029906)
  (Q22,0.005474)(Q23,0.043755)(Q24,0.031825)
  (Q25,0.093671)(Q26,0.258850)(Q27,0.137817)
  (Q28,0.028400)(Q29,0.006328)(Q30,0.005795)};
\addplot[thalbar,fill=colThal,draw=colThal!70!black,postaction={pattern=north east lines,pattern color=black}] coordinates{
  (Q16,0)(Q17,0)(Q18,0.040328)(Q19,0)(Q20,0)(Q21,0.007009)
  (Q22,0)(Q23,0.007506)(Q24,0.007084)(Q25,0)(Q26,0)(Q27,0)
  (Q28,0)(Q29,0.000153)(Q30,0)};
\addplot[flockbar,fill=colFlock,draw=colFlock!70!black,postaction={pattern=north east lines,pattern color=black}] coordinates{
  (Q16,0)(Q17,0)(Q18,0)(Q19,0)(Q20,0.000900)(Q21,0)
  (Q22,0.007800)(Q23,0)(Q24,0)(Q25,0)(Q26,0)(Q27,0)
  (Q28,0)(Q29,0.000700)(Q30,0.003800)};
\addplot[palmbar,fill=colPalm,draw=colPalm!70!black,postaction={pattern=north east lines,pattern color=black}] coordinates{
  (Q16,0.582940)(Q17,0)(Q18,0.033649)(Q19,0)(Q20,0.002173)(Q21,0.019891)
  (Q22,0.008944)(Q23,0.031007)(Q24,0.008652)(Q25,0)(Q26,0.617773)(Q27,2.213130)
  (Q28,0)(Q29,0)(Q30,0)};
\addplot[ababar,fill=colAba,draw=colAba!70!black,postaction={pattern=north east lines,pattern color=black}] coordinates{
  (Q16,0.981143)(Q17,0)(Q18,0)(Q19,0)(Q20,0.092141)(Q21,0.108323)
  (Q22,0.123041)(Q23,0)(Q24,0.125783)(Q25,0)(Q26,1.199583)(Q27,2.480959)
  (Q28,0)(Q29,0)(Q30,0)};
\addplot[failmark, draw=colThal, xshift=0pt] coordinates{
  (Q16,0.00011)(Q17,0.00011)(Q19,0.00011)(Q20,0.00011)
  (Q22,0.00011)(Q25,0.00011)(Q26,0.00011)(Q27,0.00011)(Q28,0.00011)(Q30,0.00011)};
\addplot[failmark, draw=colFlock, xshift=6pt] coordinates{
  (Q16,0.00011)(Q17,0.00011)(Q18,0.00011)(Q19,0.00011)(Q21,0.00011)
  (Q23,0.00011)(Q24,0.00011)(Q25,0.00011)(Q26,0.00011)(Q27,0.00011)(Q28,0.00011)};
\addplot[failmark, draw=colPalm, xshift=12pt] coordinates{
  (Q17,0.00011)(Q19,0.00011)
  (Q25,0.00011)(Q28,0.00011)(Q29,0.00011)(Q30,0.00011)};
\addplot[failmark, draw=colAba, xshift=18pt] coordinates{
  (Q17,0.00011)(Q18,0.00011)(Q19,0.00011)(Q23,0.00011)
  (Q25,0.00011)(Q28,0.00011)(Q29,0.00011)(Q30,0.00011)};
\addplot[failmark, draw=colCost, xshift=-12pt] coordinates{
  (Q16,0.00011)(Q25,0.00011)};
\addplot[failmark, draw=colPull, xshift=-6pt] coordinates{
  (Q19,0.00011)};
\end{axis}
\end{tikzpicture}%
}%
}

\caption{Per-query latency (s) and LLM cost (\$) on log scale across the 30-query hybrid benchmark.
  Subfigures (a, c) show latency; subfigures (b, d) show cost (hatched bars).
  Bars are shown only for methods with F1 $\geq$ 0.4 on that query; an $\times$ indicates that a method either did not meet the accuracy threshold, timed out, or ran out of memory.}
\Description{Bar charts showing per-query latency and LLM cost for all systems across the 30-query hybrid benchmark.}
\label{fig:latency-cost-all}
\end{figure*}
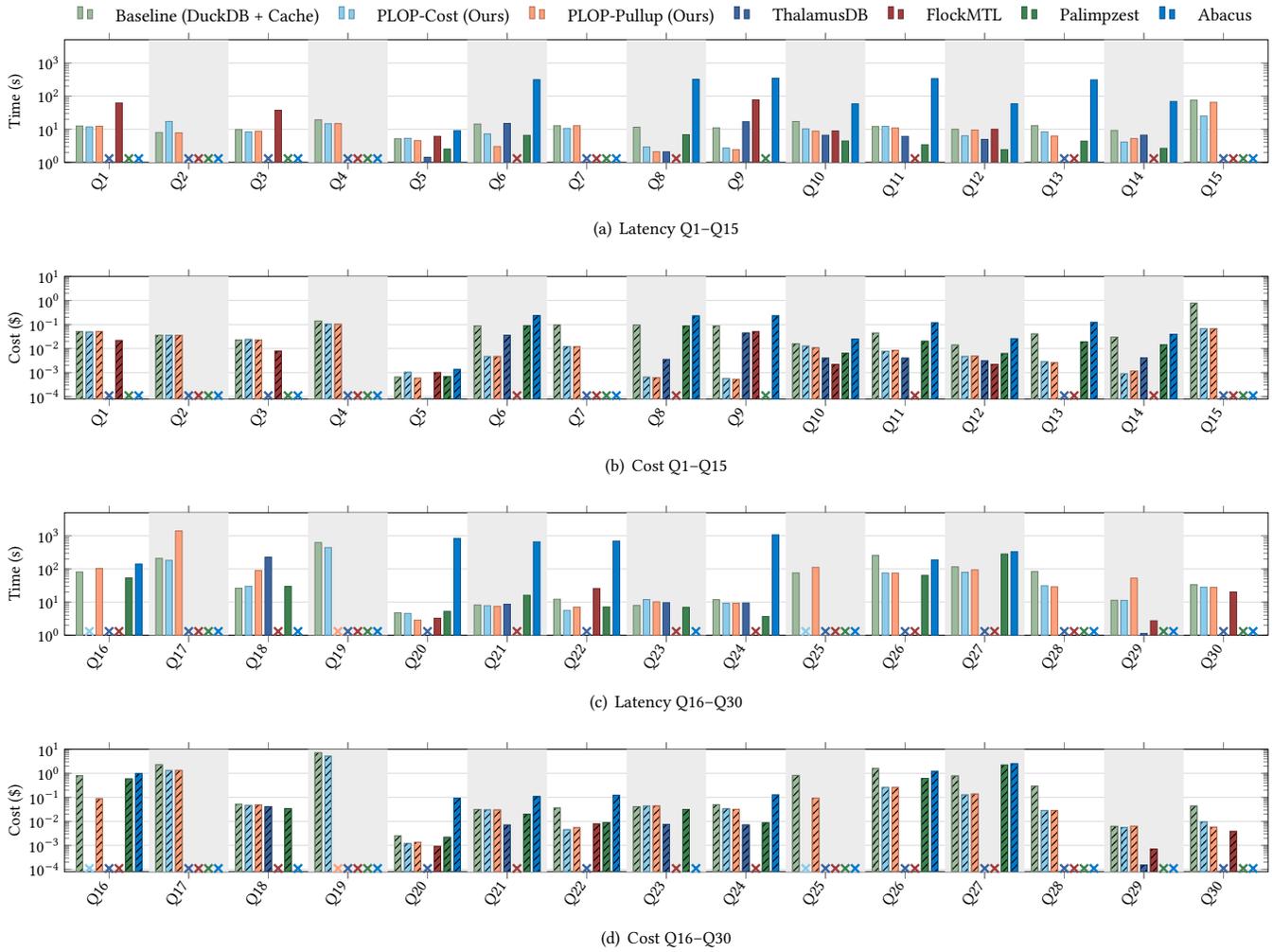

We run each system on all 30 queries and report speedup, cost reduction, and F1 scores against the DuckDB + Cache baseline.
\Cref{tab:summary} summarizes the results.
Overall, our cost model (i.e., \textsc{PLOP}-Cost) achieves $1.50\times$ speedup and $4.18\times$ cost reduction, while our pulling-up algorithm (i.e., \textsc{PLOP}-Pullup) achieves $1.16\times$ speedup and $4.29\times$ cost reduction. 
Both maintain an average F1 of $\approx 0.85$, substantially higher than all other systems.

Among the baselines, ThalamusDB achieves the highest cost reduction at $10.10\times$ but with only an F1 score of $0.52$, indicating that roughly half the results are incorrect.
FlockMTL is slower than the baseline at $0.64\times$ with the lowest accuracy at $0.34$ F1.
Palimpzest achieves the fastest execution at $1.70\times$ but only moderate cost reduction at $2.15\times$ and $0.53$ F1.
Abacus is both slower at $0.15\times$ and more expensive at $0.43\times$ than the baseline.

\Cref{fig:latency-cost-all} shows per-query latency and cost on a log scale.
\textsc{PLOP} variants consistently produce valid results across all 30 benchmark queries, while competing systems frequently produce unusable outputs on multi-table queries.
We use F1 $\geq 0.4$ as the threshold for showing a bar in \Cref{fig:latency-cost-all}: below this level, more than half the result rows are incorrect, making the output unreliable for downstream use.
Systems that fail this threshold are marked with $\times$.
The advantage of \textsc{PLOP}'s plan-level optimization is most pronounced on complex queries with 5+ joins and multiple $\SF$s, where the interplay between LLM cost and relational cost creates a significant optimization opportunity.

\subsection{SemBench E-Commerce Results}
\label{sec:exp-sembench}

\begin{table}[t]
\vspace{2ex}
\centering
\caption{SemBench E-Commerce performance (14 queries). Quality is the arithmetic mean of the official per-query benchmark scores, where different queries use different metrics (either F1 or Adjusted Rand Index). Failed or timed-out queries receive a score of 0.}
\label{tab:ecomm_perf}
\small
\renewcommand{\arraystretch}{1.15}
\resizebox{0.9\linewidth}{!}{%
\begin{tabular}{l r r r}
\toprule
\textbf{Method} & \textbf{Speedup} & \textbf{Cost Red.} & \textbf{Quality} \\
\midrule
\rowcolor{gray!12} Baseline (DuckDB + Cache) & 1.00$\times$ & 1.00$\times$ & 0.865 \\[1pt]
\rowcolor{colCost!12} \textsc{PLOP}-Cost & 1.11$\times$ & 1.04$\times$ & 0.840 \\
\rowcolor{colCost!12} \textsc{PLOP}-Pullup & 1.10$\times$ & 1.04$\times$ & 0.796 \\[1pt]
ThalamusDB$^\dagger$ & 0.40$\times$ & 0.30$\times$ & 0.140 \\
Palimpzest & 2.26$\times$ & 0.29$\times$ & 0.500 \\
Lotus & 2.27$\times$ & 1.67$\times$ & 0.574 \\
Abacus & 0.67$\times$ & 0.48$\times$ & 0.309 \\
\bottomrule
\end{tabular}%
}
\par\vspace{2pt}

{\footnotesize $^\dagger$\, ThalamusDB is not available for q3, q5, q6, q10--q14 in SemBench due to the lack of support for cross-modal semantic join and $\SP$~\cite{lao2026sembenchbenchmarksemanticquery}  } 
\end{table}

The F1 scores on the 30-query benchmark are measured by comparing with the baseline output. To validate accuracy against human annotations, we evaluate on SemBench E-Commerce (\Cref{tab:ecomm_perf}).
\textsc{PLOP}-Cost achieves $0.840$ quality score defined by SemBench, close to the baseline's score of $0.865$, with a speedup of $1.11\times$.
These queries are simple with few joins, leaving little room for placement optimization. 
The observed gains mainly come from pulling up semantic filters through traditional relational operators (e.g., joins), leveraging function caching to reduce redundant LLM invocations.
The small gains confirm that \textsc{PLOP} preserves accuracy on simple queries.
Lotus and Palimpzest are faster but return lower quality results, while ThalamusDB and Abacus fall below the baseline on all metrics.

\subsection{Pull-up vs.\ Cost Model}
\label{sec:exp-pullup-vs-cost}

\begin{figure}[t]
\vspace{10pt}
\begin{lstlisting}[
  language=SQL, basicstyle=\ttfamily\small, frame=tb, framerule=0.4pt,
  rulecolor=\color{black}, breaklines=true, showstringspaces=false,
  morekeywords={SEMANTIC}, keywordstyle=\bfseries, stringstyle=\itshape,
  numbers=left, numberstyle=\tiny\color{gray}, numbersep=6pt,
  xleftmargin=14pt, xrightmargin=2pt, aboveskip=4pt, belowskip=4pt,
  captionpos=b, abovecaptionskip=8pt,
  caption={Q8: Semantic book lookup with filtered reviews (abbreviated).},
  label=lst:q8,
]
WITH candidates AS (
  SELECT CAST(SPLIT_PART(b.book_id,'_',2) AS INTEGER)
           AS book_idx, b.book_id, b.title
  FROM books_info b
  WHERE b.title IS NOT NULL
    AND SEMANTIC('Confirm this is the second edition of
      Make: Electronics ... Title: {b.title}
      Subtitle: {b.subtitle} Author: {b.author}
      Categories: {b.categories}')
),
reviews_filtered AS (
  SELECT * FROM reviews r
  WHERE r.verified_purchase = 1
    AND CAST(r.rating AS DOUBLE) >= 5
    AND r.helpful_vote >= 50
    AND r.review_time >= TIMESTAMP '2017-01-01'
    AND r.review_time <  TIMESTAMP '2018-01-01'
)
SELECT * FROM candidates c
JOIN reviews_filtered rf
  ON rf.purchase_idx = c.book_idx
ORDER BY rf.review_time DESC;
\end{lstlisting}
\end{figure}

\begin{figure}[t]
\vspace{10pt}
\begin{lstlisting}[
  language=SQL, basicstyle=\ttfamily\small, frame=tb, framerule=0.4pt,
  rulecolor=\color{black}, breaklines=true, showstringspaces=false,
  morekeywords={SEMANTIC}, keywordstyle=\bfseries, stringstyle=\itshape,
  numbers=left, numberstyle=\tiny\color{gray}, numbersep=6pt,
  xleftmargin=14pt, xrightmargin=2pt, aboveskip=4pt, belowskip=4pt,
  captionpos=b, abovecaptionskip=8pt,
  caption={Q19: Multi-join audit query with two \texttt{SEMANTIC} filters (abbreviated).},
  label=lst:q19,
]
WITH lineitem_returns AS (
  SELECT * FROM lineitem l
  WHERE l.l_shipdate BETWEEN DATE '1994-01-01'
                         AND DATE '1998-01-01'
    AND l.l_returnflag IN ('R','A','N')
    AND l.l_quantity BETWEEN 3 AND 38
    AND SEMANTIC('Mode: {l.l_shipmode}
      Instruction: {l.l_shipinstruct}
      Is this a potentially problematic fulfillment
      case? Answer YES or NO.')
),
order_customer AS (
  SELECT * FROM orders o
  JOIN customer c ON c.c_custkey = o.o_custkey
  WHERE o.o_orderdate BETWEEN DATE '1994-01-01'
                          AND DATE '1998-01-01'
    AND o.o_orderstatus IN ('O','F')
    AND o.o_totalprice > 20000
),
part_supplier AS (
  SELECT * FROM part p
  JOIN partsupp ps ON ps.ps_partkey = p.p_partkey
  JOIN supplier s  ON s.s_suppkey  = ps.ps_suppkey
  WHERE p.p_size BETWEEN 1 AND 40
),
customer_context AS (
  SELECT * FROM customer_sample c
  WHERE SEMANTIC('Segment: {c.c_mktsegment}
      Balance: {c.c_acctbal}
      Higher complaint/escalation risk? YES or NO.')
)
SELECT * FROM lineitem_returns lr
JOIN order_customer  oc ON oc.o_orderkey = lr.l_orderkey
JOIN part_supplier   ps ON ps.p_partkey  = lr.l_partkey
                       AND ps.s_suppkey  = lr.l_suppkey
CROSS JOIN customer_context cc;
\end{lstlisting}
\end{figure}

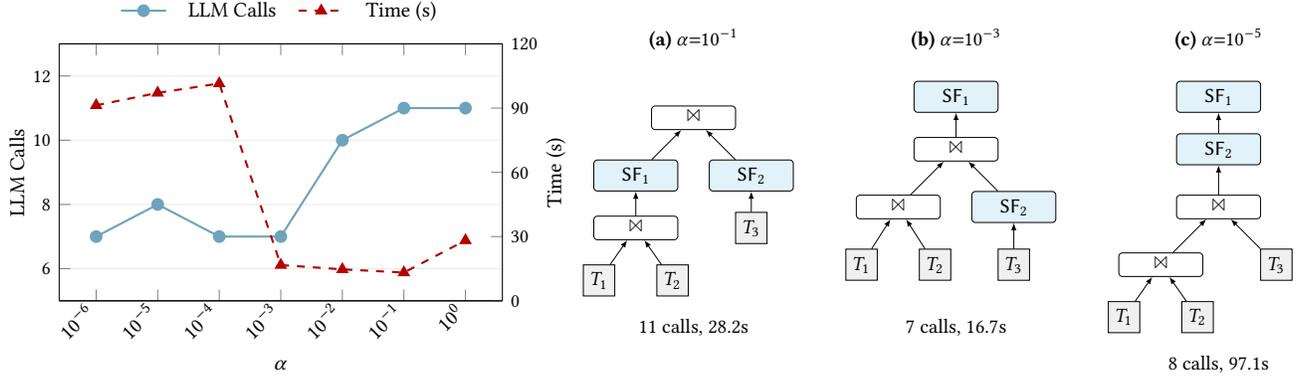
\begin{figure*}[t]
\centering
\begin{minipage}[b]{0.42\linewidth}
\centering
\begin{tikzpicture}
\pgfplotsset{every tick label/.append style={font=\footnotesize}}
\begin{axis}[
  name=ax1,
  width=\linewidth,
  height=5.0cm,
  xlabel={$\alpha$},
  xlabel style={font=\small},
  ylabel={LLM Calls},
  ylabel style={font=\small, yshift=-3pt},
  symbolic x coords={$10^{-6}$,$10^{-5}$,$10^{-4}$,$10^{-3}$,$10^{-2}$,$10^{-1}$,$10^{0}$},
  xtick=data,
  xticklabel style={font=\small, rotate=45, anchor=east},
  ymin=5, ymax=13,
  ytick={6,8,10,12},
  axis y line*=left,
  ymajorgrids=true,
  grid style={draw=gray!15},
  legend style={
    font=\small,
    at={(0.5,1.05)},
    anchor=south,
    draw=none,
    legend columns=2,
    column sep=6pt,
  },
  mark size=2pt,
]
\addplot[mark=*, thick, colCost!80!black, mark options={solid}] coordinates{
  ($10^{-6}$, 7) ($10^{-5}$, 8) ($10^{-4}$, 7) ($10^{-3}$, 7) ($10^{-2}$, 10) ($10^{-1}$, 11) ($10^{0}$, 11)};
\addlegendentry{LLM Calls}
\addplot[mark=triangle*, thick, red!70!black, dashed, mark options={solid}] coordinates{
  ($10^{-6}$, 0) ($10^{-5}$, 0) ($10^{-4}$, 0) ($10^{-3}$, 0) ($10^{-2}$, 0) ($10^{-1}$, 0) ($10^{0}$, 0)};
\addlegendentry{Time (s)}
\end{axis}

\begin{axis}[
  width=\linewidth,
  height=5.0cm,
  symbolic x coords={$10^{-6}$,$10^{-5}$,$10^{-4}$,$10^{-3}$,$10^{-2}$,$10^{-1}$,$10^{0}$},
  xtick=\empty,
  axis y line*=right,
  ylabel={Time (s)},
  ylabel style={font=\small, yshift=3pt},
  yticklabel style={font=\footnotesize},
  ymin=0, ymax=120,
  ytick={0,30,60,90,120},
  mark size=2pt,
]
\addplot[mark=triangle*, thick, red!70!black, dashed, mark options={solid}] coordinates{
  ($10^{-6}$, 91.3) ($10^{-5}$, 97.1) ($10^{-4}$, 101.5) ($10^{-3}$, 16.7) ($10^{-2}$, 14.7) ($10^{-1}$, 13.2) ($10^{0}$, 28.2)};
\end{axis}
\end{tikzpicture}
\end{minipage}%
\hspace{5pt}
\begin{minipage}[b]{0.53\linewidth}
\centering
\resizebox{\linewidth}{!}{%
\begin{tikzpicture}[
    op/.style={draw, rounded corners=2pt, font=\Large, inner sep=4pt, minimum width=1.6cm, align=center},
    sem/.style={op, fill=colCost!25},
    tbl/.style={draw, fill=gray!12, font=\Large, inner sep=4pt, align=center},
    lbl/.style={font=\Large\bfseries},
    ann/.style={font=\Large, align=center},
    every edge/.style={draw, -latex, thick},
]

\node[lbl] at (0, 4.5) {(a) $\alpha{=}10^{-1}$};
\node[op]  (A-jt)  at (0, 3.0) {$\Join$};
\node[sem] (A-sf1) at (-1.1, 1.9) {$\SF_1$};
\node[op]  (A-jl)  at (-1.1, 0.9) {$\Join$};
\node[tbl] (A-t1)  at (-1.8, -0.1) {$T_1$};
\node[tbl] (A-t2)  at (-0.4, -0.1) {$T_2$};
\node[sem] (A-sf2) at (1.1, 1.9) {$\SF_2$};
\node[tbl] (A-t3)  at (1.1, 0.9) {$T_3$};
\draw[-latex] (A-t1) -- (A-jl);
\draw[-latex] (A-t2) -- (A-jl);
\draw[-latex] (A-jl) -- (A-sf1);
\draw[-latex] (A-sf1) -- (A-jt);
\draw[-latex] (A-t3) -- (A-sf2);
\draw[-latex] (A-sf2) -- (A-jt);
\node[ann] at (0, -1.0) {11 calls, 28.2s};

\node[lbl] at (5.0, 4.5) {(b) $\alpha{=}10^{-3}$};
\node[sem] (B-sf1) at (5.0, 3.4) {$\SF_1$};
\node[op]  (B-jt)  at (5.0, 2.4) {$\Join$};
\node[op]  (B-jl)  at (3.9, 1.3) {$\Join$};
\node[tbl] (B-t1)  at (3.2, 0.2) {$T_1$};
\node[tbl] (B-t2)  at (4.6, 0.2) {$T_2$};
\node[sem] (B-sf2) at (6.1, 1.3) {$\SF_2$};
\node[tbl] (B-t3)  at (6.1, 0.2) {$T_3$};
\draw[-latex] (B-t1) -- (B-jl);
\draw[-latex] (B-t2) -- (B-jl);
\draw[-latex] (B-jl) -- (B-jt);
\draw[-latex] (B-t3) -- (B-sf2);
\draw[-latex] (B-sf2) -- (B-jt);
\draw[-latex] (B-jt) -- (B-sf1);
\node[ann] at (5.0, -1.0) {7 calls, 16.7s};

\node[lbl] at (10.0, 4.5) {(c) $\alpha{=}10^{-5}$};
\node[sem] (C-sf1) at (10.0, 3.4) {$\SF_1$};
\node[sem] (C-sf2) at (10.0, 2.4) {$\SF_2$};
\node[op]  (C-jt)  at (10.0, 1.3) {$\Join$};
\node[op]  (C-jl)  at (8.9, 0.2) {$\Join$};
\node[tbl] (C-t1)  at (8.2, -0.8) {$T_1$};
\node[tbl] (C-t2)  at (9.6, -0.8) {$T_2$};
\node[tbl] (C-t3)  at (11.1, 0.2) {$T_3$};
\draw[-latex] (C-t1) -- (C-jl);
\draw[-latex] (C-t2) -- (C-jl);
\draw[-latex] (C-jl) -- (C-jt);
\draw[-latex] (C-t3) -- (C-jt);
\draw[-latex] (C-jt) -- (C-sf2);
\draw[-latex] (C-sf2) -- (C-sf1);
\node[ann] at (10.0, -1.7) {8 calls, 97.1s};

\end{tikzpicture}%
}
\end{minipage}
\caption{Sensitivity to $\alpha$. Left: LLM calls and latency as $\alpha$ varies. Right: plan trees under three settings.
(a)~Large $\alpha$: both filters pushed down, 11 calls but 28.2s.
(b)~Moderate $\alpha$: $\SF_1$ pulled above the top join, 7 calls at 16.7s.
(c)~Small $\alpha$: both filters pulled up, 8 calls but 97.1s as unfiltered joins dominate.}
\Description{Left: dual-axis line plot of LLM calls and time vs alpha. Right: three plan trees under different alpha.}
\label{fig:alpha-sensitivity}
\end{figure*}

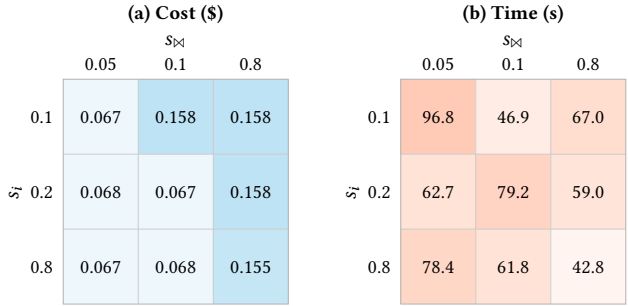
\begin{figure}[t]
\centering
\resizebox{1\columnwidth}{!}{%
\begin{tikzpicture}[
  hcell/.style={minimum width=1.2cm, minimum height=1.2cm, align=center, font=\normalsize, inner sep=0pt},
  hdr/.style={font=\normalsize},
  ttl/.style={font=\normalsize\bfseries},
]

\path[use as bounding box] (-1.0, -0.6) rectangle (9.2, 4.3);


\node[ttl] at (1.8, 3.9) {(a) Cost (\$)};
\node[hdr] at (1.8, 3.45) {{\large $s_{\Join}$}};
\node[hdr] at (0.6, 3.1) {0.05};
\node[hdr] at (1.8, 3.1) {0.1};
\node[hdr] at (3.0, 3.1) {0.8};
\node[hdr, rotate=90, anchor=south] at (-0.55, 1.05) {{\large $s_i$}};
\node[hdr, anchor=east] at (-0.05, 2.25) {0.1};
\node[hdr, anchor=east] at (-0.05, 1.05) {0.2};
\node[hdr, anchor=east] at (-0.05, -0.15) {0.8};

\node[hcell, fill=colCost!15] at (0.6, 2.25) {0.067};
\node[hcell, fill=colCost!55] at (1.8, 2.25) {0.158};
\node[hcell, fill=colCost!55] at (3.0, 2.25) {0.158};
\node[hcell, fill=colCost!15] at (0.6, 1.05) {0.068};
\node[hcell, fill=colCost!15] at (1.8, 1.05) {0.067};
\node[hcell, fill=colCost!55] at (3.0, 1.05) {0.158};
\node[hcell, fill=colCost!15] at (0.6, -0.15) {0.067};
\node[hcell, fill=colCost!15] at (1.8, -0.15) {0.068};
\node[hcell, fill=colCost!50] at (3.0, -0.15) {0.155};
\draw[gray!40] (0.0,0.45) -- (3.6,0.45);
\draw[gray!40] (0.0,1.65) -- (3.6,1.65);
\draw[gray!40] (1.2,-0.75) -- (1.2,2.85);
\draw[gray!40] (2.4,-0.75) -- (2.4,2.85);
\draw[gray!60] (0.0,-0.75) rectangle (3.6,2.85);

\def\ox{5.4}
\node[ttl] at (\ox+1.8, 3.9) {(b) Time (s)};
\node[hdr] at (\ox+1.8, 3.45) {{\large $s_{\Join}$}};
\node[hdr] at (\ox+0.6, 3.1) {0.05};
\node[hdr] at (\ox+1.8, 3.1) {0.1};
\node[hdr] at (\ox+3.0, 3.1) {0.8};
\node[hdr, rotate=90, anchor=south] at (\ox-0.55, 1.05) {{\large $s_i$}};
\node[hdr, anchor=east] at (\ox-0.05, 2.25) {0.1};
\node[hdr, anchor=east] at (\ox-0.05, 1.05) {0.2};
\node[hdr, anchor=east] at (\ox-0.05, -0.15) {0.8};

\node[hcell, fill=colPull!55] at (\ox+0.6, 2.25) {96.8};
\node[hcell, fill=colPull!20] at (\ox+1.8, 2.25) {46.9};
\node[hcell, fill=colPull!35] at (\ox+3.0, 2.25) {67.0};
\node[hcell, fill=colPull!30] at (\ox+0.6, 1.05) {62.7};
\node[hcell, fill=colPull!45] at (\ox+1.8, 1.05) {79.2};
\node[hcell, fill=colPull!28] at (\ox+3.0, 1.05) {59.0};
\node[hcell, fill=colPull!42] at (\ox+0.6, -0.15) {78.4};
\node[hcell, fill=colPull!30] at (\ox+1.8, -0.15) {61.8};
\node[hcell, fill=colPull!12] at (\ox+3.0, -0.15) {42.8};
\draw[gray!40] (\ox+0.0,0.45) -- (\ox+3.6,0.45);
\draw[gray!40] (\ox+0.0,1.65) -- (\ox+3.6,1.65);
\draw[gray!40] (\ox+1.2,-0.75) -- (\ox+1.2,2.85);
\draw[gray!40] (\ox+2.4,-0.75) -- (\ox+2.4,2.85);
\draw[gray!60] (\ox+0.0,-0.75) rectangle (\ox+3.6,2.85);

\end{tikzpicture}%
}
\caption{Sensitivity to selectivity estimates. (a) LLM cost and (b) wall-clock time as $s_i$ and $s_{\Join}$ vary. The DP cost model produces two plans: low cost when $s_{\Join} \leq 0.05$, higher cost when $s_{\Join}$ increases.}
\Description{Two heatmaps showing cost and time sensitivity to selectivity parameters.}
\label{fig:selectivity-heatmap}
\end{figure}

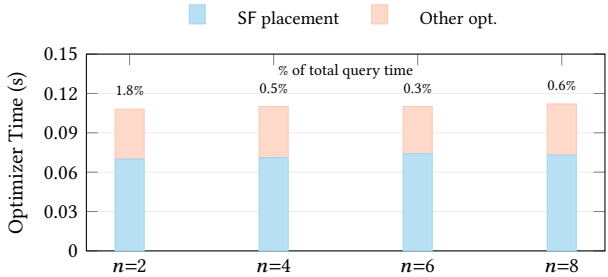
\begin{figure}[t]
\centering
\begin{tikzpicture}
\begin{axis}[
  ybar stacked,
  width=1\columnwidth,
  height=4.2cm,
  bar width=11pt,
  symbolic x coords={$n{=}2$, $n{=}4$, $n{=}6$, $n{=}8$},
  xtick=data,
  xticklabel style={font=\small},
  yticklabel style={font=\small},
  ylabel={Optimizer Time (s)},
  ylabel style={font=\small},
  ymin=0, ymax=0.15,
  ytick={0, 0.03, 0.06, 0.09, 0.12, 0.15},
  yticklabels={0, 0.03, 0.06, 0.09, 0.12, 0.15},
  ymajorgrids=true,
  grid style={draw=gray!15},
  legend style={
    font=\footnotesize,
    at={(0.5,1.08)},
    anchor=south,
    draw=none,
    legend columns=2,
    column sep=10pt,
  },
  clip=false,
]

\addplot[fill=colCost!60, draw=colCost!80] coordinates{
  ($n{=}2$, 0.070) ($n{=}4$, 0.071) ($n{=}6$, 0.074) ($n{=}8$, 0.073)};

\addplot[fill=colPull!40, draw=colPull!60] coordinates{
  ($n{=}2$, 0.038) ($n{=}4$, 0.039) ($n{=}6$, 0.036) ($n{=}8$, 0.039)};

\legend{$\SF$ placement, Other opt.}

\node[font=\scriptsize, anchor=south] at (axis cs:$n{=}2$, 0.112) {1.8\%};
\node[font=\scriptsize, anchor=south] at (axis cs:$n{=}4$, 0.114) {0.5\%};
\node[font=\scriptsize, anchor=south] at (axis cs:$n{=}6$, 0.114) {0.3\%};
\node[font=\scriptsize, anchor=south] at (axis cs:$n{=}8$, 0.116) {0.6\%};
\node[font=\scriptsize, anchor=south] at (rel axis cs:0.5, 0.83) {\% of total query time};

\end{axis}
\end{tikzpicture}
\caption{Optimizer overhead by number of $\SF$s. Percentages above bars show the optimizer's share of total query execution time. The optimizer takes under 0.12s in all cases.}
\Description{Stacked bar chart showing optimizer overhead.}
\label{fig:optimizer-overhead}
\end{figure}

PLOP's two optimization strategies target different query regimes.
We illustrate with two concrete examples.

For simple queries with few joins and one $\SF$, pull-up alone suffices.
Q8 from BookReview (Listing~\ref{lst:q8}) searches for a specific book edition using one $\SF$ over a single join between \texttt{books\_info} and \texttt{reviews}.
With only 1 join, pulling up the $\SF$ adds negligible relational cost while reducing LLM calls from 200 to 1 via pre-filtering tuples by relational filters and the inner-join
Here, \text{PLOP}'s two optimization strategies take a similar time (pull-up takes 2.1s vs.\ the cost model's 2.9s). 

For complex multi-table queries, the two strategies diverge significantly.
Q19 from TPC-H (Listing~\ref{lst:q19}) is a multi-table audit query with 6 joins, 2 $\SF$s, and 4 relational filters spanning \texttt{lineitem}, \texttt{orders}, \texttt{customer}, \texttt{part}, \texttt{partsupp}, and \texttt{supplier}.
Pulling up both $\SF$s forces all 6 joins to process unfiltered tables, causing intermediate results to explode.
Pull-up reaches the 3{,}000s timeout while the cost model finishes in 443s by keeping the one $\SF$ not pulled up, shrinking the join inputs early.
In summary, pull-up is the better strategy for simple queries with few joins, while the cost model is essential for complex multi-table queries where unfiltered joins cause latency to increase.

\subsection{Ablations}
\label{sec:exp-ablations}

\paragraph{Sensitivity to $\alpha$.}
The parameter $\alpha$ scales the relational cost term in $C_{\mathrm{LLM}} + \alpha \cdot C_{\mathrm{rel}}$.
A large $\alpha$ penalizes relational cost heavily, so the DP algorithm pushes filters down to shrink join inputs at the expense of more LLM calls.
A small $\alpha$ deprioritizes relational cost, so the DP algorithm pulls filters up to minimize LLM calls even if joins grow larger.
\Cref{fig:alpha-sensitivity} shows how varying $\alpha$ affects plan quality on a representative multi-table query.
The cost model is robust across several orders of magnitude.
For $\alpha \in [10^{-3}, 10^{0}]$, the number of LLM calls remain stable at 7--11 and latency stays under 30s.
At small $\alpha$ values from $10^{-4}$ to $10^{-6}$, the DP algorithm increasingly pulls filters up, reducing LLM calls to 7--8 but causing latency to spike above 90s as unfiltered joins dominate execution time.
Our default $\alpha = 10^{-7}$ lies below the plotted range but in the same pull-up-favoring regime; the plan at $10^{-7}$ is identical to the plan at $10^{-6}$.
We use such a small default $\alpha$ because our primary objective is to optimize monetary cost from LLM calls, while still discouraging plans whose relational cost blowup would make the query prohibitively slow:
on simple queries, the resulting plans are near-optimal, while on complex multi-table queries, the DP algorithm still selectively pushes down filters when the relational cost blowup outweighs the LLM savings. 


\paragraph{Sensitivity to selectivity estimates.}
As described in \Cref{sec:implementation}, $N_{u,\SF_i}$ is estimated by multiplying the base table size by $s_i$ at each semantic filter and $s_{\Join}$ at each join along the path.
Since $N_{u,\SF_i}$ directly determines the LLM cost of placing a filter at $u$, errors in these estimates can shift the balance between pulling up and pushing down.
To evaluate robustness, we vary $s_i$ and $s_{\Join}$ independently on a representative multi-table query.
\Cref{fig:selectivity-heatmap} shows cost and time as heatmaps.
Our DP algorithm produces two distinct plans depending on $s_{\Join}$.
When $s_{\Join} \leq 0.05$, the estimated join reduction is large, so the algorithm keeps filters pulled up, yielding $\sim$211 LLM calls costing \$0.07.
When $s_{\Join} \geq 0.1$, pushing down becomes more attractive, increasing calls to $\sim$510 costing \$0.16.
In general, cost and latency remain similar across all $s_i$ and $s_{\Join}$ values within each regime, indicating that the placement is robust to moderate selectivity errors.

\paragraph{Optimizer overhead.}
\textsc{PLOP}'s DP algorithm runs in $\mathcal{O}(|V| \cdot n \cdot 2^n + 3^n)$ time.
\Cref{fig:optimizer-overhead} shows the optimizer latency for queries grouped by the number of $\SF$s at $n = 2, 4, 6, 8$.
We decompose the total optimizer time into the $\SF$ placement, which runs the DP algorithm from \Cref{alg:dp}, and the remaining overhead from $\SP$ pull-up, $\SJ$ decomposition, and plan rewriting.
The $\SF$ placement accounts for about 65\% of the total optimizer time, with the DP algorithm itself taking under 0.08s even for $n = 8$.
The near-constant time across $n = 2$ to $8$ reflects that plan traversal and rewriting dominate; the $3^n$ term becomes significant only for larger $n$.
In practice, $n = 8$ already represents an extreme case: production semantic workloads at Snowflake report that most queries contain fewer than 4 semantic operators~\cite{liskowski2025cortex}.
The total optimizer overhead stays below 2\% of end-to-end query time in all cases, as shown by the annotations in \Cref{fig:optimizer-overhead}.
We report the overhead only for the cost model. As shown in \Cref{alg:pullup}, the pull-up algorithm runs in $\mathcal{O}(n^2 d)$ with negligible cost.

\section{Discussion}
\label{sec:discussion}

\paragraph{Accuracy analysis.}
On the hybrid query benchmark in \Cref{sec:experiments}, both \textsc{PLOP} strategies achieve an average F1 score of around $0.85$, with only a small gap from the unoptimized DuckDB + Cache reference.
Since \textsc{PLOP} only changes operator placement without modifying individual operators, this gap is unlikely to arise from semantic changes introduced by plan rewriting.
A more plausible explanation is LLM non-determinism: evaluating the same semantic predicate in separate executions can yield different outputs on borderline cases.
This interpretation is consistent with our results on SemBench, where human-annotated ground truth is available and \textsc{PLOP} achieves quality on par with DuckDB + Cache.
Together, these results suggest that our placement algorithms do not introduce systematic quality degradation.

\paragraph{Comparison with classical expensive predicate optimization.}
Classical work on expensive predicate placement does not fully address the cost characteristics of semantic operators.
Prior work~\cite{hellerstein1993predicate} introduced predicate migration with function caching under the assumption of linear join cost.
Semantic joins break this assumption: $\SJ$ evaluates each pair of tuples with an LLM call, resulting in $|R| \times |S|$ calls in the worst case, a quadratic cost that the linear model cannot capture.
\textsc{PLOP} does not rely on this assumption; its DP cost model directly estimates LLM and relational costs at each node across hybrid plans using $N_{u,\SF_i}$ and selectivities, handling non-linear cost structures naturally.
Another work~\cite{chaudhuri1999optimization} introduced a DP-based cost model for ordering expensive predicates, but their formulation does not account for function caching.
Function caching is critical for semantic queries as it reduces the effective LLM cost from total rows to distinct rows, fundamentally changing which placements are optimal.
\textsc{PLOP} combines function-caching-aware cost modeling with a DP algorithm that jointly balances LLM and relational costs across multi-table plans.


\section{Related Work}
\label{sec:related}

\paragraph{UDF optimization.}
Hellerstein and Stonebraker~\cite{hellerstein1993predicate} introduced predicate migration and function caching for expensive user-defined predicates, enabling the optimizer to reposition UDFs within query plans under a linear join cost assumption.
Chaudhuri and Shim~\cite{chaudhuri1999optimization} extended this to a DP-based cost model for ordering expensive predicates, but their formulation does not account for function caching.
GRACEFUL~\cite{wehrstein2025graceful} addresses UDF placement with a learned GNN-based cost estimator, achieving significant speedups through informed pull-up and push-down decisions.
Chasialis et al.~\cite{chasialis2025optimizing} study UDF query optimization in SQL data engines with operator fusion and pluggable registration.
Other recent work further improves UDF execution: Arch et al.~\cite{arch2024key} propose outlining UDFs before inlining to expose optimization opportunities, and Franz et al.~\cite{franz2024dear} show that batching UDF invocations outperforms inlining in many cases.
LLM-backed semantic operators violate the assumptions underlying these frameworks: semantic joins can incur quadratic LLM cost in the worst case, while function caching makes cost depend on the number of distinct inputs rather than the total number of rows.
\textsc{PLOP} builds on these ideas by combining function caching~\cite{hellerstein1993predicate} with a dynamic-programming-based optimization algorithm in the spirit of Chaudhuri and Shim~\cite{chaudhuri1999optimization}, while introducing a caching-aware cost model that is absent from prior work.

\paragraph{Semantic query systems.}
Several systems have been proposed to optimize semantic data processing.
Prior work spans operator-level, execution-level, and plan-level optimization.
Lotus~\cite{patel2025semantic} introduces semantic operators and multiple optimization strategies for them.
Palimpzest~\cite{liu2025palimpzest} and Abacus~\cite{russo2025abacuscostbasedoptimizersemantic} focus on cost-based optimization over models, prompts, and alternative implementations of semantic operators.
FlockMTL~\cite{dorbani2025beyond}, ThalamusDB~\cite{jo2024thalamusdb}, and OPHR~\cite{liu2025optimizing} improve execution efficiency through system-level techniques such as batching, approximate query processing, and KV-cache-aware request organization.
ZenDB~\cite{lin2025zendb} and ScaleDoc~\cite{zhang2025scaledoc} target large document collections using structure-aware execution or offline representations with cascaded filtering.
DocETL~\cite{shankar2024docetl} is closest to plan-level optimization, but for unstructured document processing pipelines rather than relational query optimization: it performs logical rewrites, plan evaluation, and search over rewritten pipelines.
For relational semantic queries, iPDB~\cite{kumarasinghe2026ipdb}, Cortex AISQL~\cite{liskowski2025cortex}, and Sema~\cite{qi2026sema} integrate semantic operators into SQL; Sema further supports adaptive reordering and relational-constraint deduction.
Nirvana~\cite{zhu2025beyond} performs LLM-driven search over semantically equivalent plans.
\textsc{PLOP} differs by focusing specifically on semantics-preserving placement of relational semantic filters under function caching.
These directions are complementary: operator-level and execution-level optimizations can be combined with \textsc{PLOP}'s plan-level placement decisions.

\paragraph{LLMs for database systems.}
A separate line of work applies LLMs to improve database systems themselves, spanning text-to-SQL~\cite{hong2024text2sql}, query optimization~\cite{li2024llmr2, tan2025llmqo}, SQL dialect translation~\cite{zhou2025cracksql}, automated tuning~\cite{giannakouris2025lambda, zhou2024dbgpt}, and DBMS testing~\cite{mang2026argus}.
For query optimization, LLM-R2~\cite{li2024llmr2} uses LLMs to guide rule-based query rewriting, while Tan et al.~\cite{tan2025llmqo} explore whether LLMs can serve as query optimizers for relational databases.
GALOIS~\cite{satriani2025logical} takes a different approach by treating the LLM as a data source and introducing logical and physical optimizations for executing SQL queries over LLMs.
These approaches are orthogonal to \textsc{PLOP}: they use LLMs to improve database internals, while we optimize the execution of LLM-backed operators within query plans.
The two approaches are complementary: an LLM-based query optimizer could determine join orders, while \textsc{PLOP} optimizes semantic filter placement for the resulting query plans.

\section{Conclusion}
\label{sec:conclusion}
We presented \textsc{PLOP}, a plan-level optimizer that determines where to place semantic operators in hybrid query plans to minimize the combined LLM and relational execution cost.
\textsc{PLOP} reduces the problem to semantic filter placement, proves that pull-up with function caching minimizes LLM invocations, and uses a DP-based cost model to balance LLM and relational costs when pull-up alone is insufficient.
On 44 queries across five schemas and two benchmarks, \textsc{PLOP} achieves up to $1.5\times$ speedup and $4.3\times$ cost reduction while maintaining the highest accuracy among six evaluated systems.
Compared to existing semantic query systems that modify operator internals, \textsc{PLOP} preserves operator outputs and optimizes placement across the entire plan tree.
\textsc{PLOP}'s plan-level approach is orthogonal to operator-level techniques and can serve as a foundation for cost-aware hybrid query processing as semantic operators become standard in data systems.


\clearpage

\balance
\bibliographystyle{ACM-Reference-Format}
\bibliography{References}

@article{patel2025semantic,
  title={Semantic Operators and Their Optimization: Enabling LLM-Based Data Processing with Accuracy Guarantees in LOTUS},
  author={Patel, Liana and Jha, Siddharth and Pan, Melissa and Gupta, Harshit and Asawa, Parth and Guestrin, Carlos and Zaharia, Matei},
  journal={Proceedings of the VLDB Endowment},
  volume={18},
  number={11},
  pages={4171--4184},
  year={2025},
  publisher={VLDB Endowment}
}

@article{kang2020approximate,
  title={Approximate selection with guarantees using proxies},
  author={Kang, Daniel and Gan, Edward and Bailis, Peter and Hashimoto, Tatsunori and Zaharia, Matei},
  journal={arXiv preprint arXiv:2004.00827},
  year={2020}
}

@inproceedings{liu2025palimpzest,
  title={Palimpzest: Optimizing ai-powered analytics with declarative query processing},
  author={Liu, Chunwei and Russo, Matthew and Cafarella, Michael and Cao, Lei and Chen, Peter Baile and Chen, Zui and Franklin, Michael and Kraska, Tim and Madden, Samuel and Shahout, Rana and others},
  booktitle={Proceedings of the Conference on Innovative Database Research (CIDR)},
  pages={2},
  year={2025}
}

@misc{russo2025abacuscostbasedoptimizersemantic,
      title={Abacus: A Cost-Based Optimizer for Semantic Operator Systems}, 
      author={Matthew Russo and Sivaprasad Sudhir and Gerardo Vitagliano and Chunwei Liu and Tim Kraska and Samuel Madden and Michael Cafarella},
      year={2025},
      eprint={2505.14661},
      archivePrefix={arXiv},
      primaryClass={cs.DB},
      url={https://arxiv.org/abs/2505.14661}, 
}

@article{dorbani2025beyond,
  title={Beyond quacking: Deep integration of language models and RAG into DuckDB},
  author={Dorbani, Anas and Yasser, Sunny and Lin, Jimmy and Mhedhbi, Amine},
  journal={arXiv preprint arXiv:2504.01157},
  year={2025}
}

@article{jo2024thalamusdb,
  title={Thalamusdb: Approximate query processing on multi-modal data},
  author={Jo, Saehan and Trummer, Immanuel},
  journal={Proceedings of the ACM on Management of Data},
  volume={2},
  number={3},
  pages={1--26},
  year={2024},
  publisher={ACM New York, NY, USA}
}

@article{kumarasinghe2026ipdb,
  title={iPDB--Optimizing SQL Queries with ML and LLM Predicates},
  author={Kumarasinghe, Udesh and Liu, Tyler and Liu, Chunwei and Aref, Walid G},
  journal={arXiv preprint arXiv:2601.16432},
  year={2026}
}

@misc{qi2026sema,
  title={Sema: A High-performance System for LLM-based Semantic Query Processing},
  author={Kangkang Qi and Dongyang Xie and Wenbo Li and Hao Zhang and Yuanyuan Zhu and Jeffrey Xu Yu and Kangfei Zhao},
  year={2026},
  eprint={2603.11622},
  archivePrefix={arXiv},
  primaryClass={cs.DB}
}

@article{shankar2024docetl,
  title={Docetl: Agentic query rewriting and evaluation for complex document processing},
  author={Shankar, Shreya and Chambers, Tristan and Shah, Tarak and Parameswaran, Aditya G and Wu, Eugene},
  journal={arXiv preprint arXiv:2410.12189},
  year={2024}
}

@inproceedings{lin2025zendb,
  title={Towards Accurate and Efficient Document Analytics with Large Language Models},
  author={Yiming Lin and Madelon Hulsebos and Ruiying Ma and Shreya Shankar and Sepanta Zeighami and Aditya G. Parameswaran},
  booktitle={Proceedings of the IEEE International Conference on Data Engineering (ICDE)},
  year={2025}
}

@misc{liskowski2025cortex,
  title={Cortex AISQL: A Production SQL Engine for Unstructured Data},
  author={Pawe{\l} Liskowski and Benjamin Han and Paritosh Aggarwal and Bowei Chen and Boxin Jiang and Nitish Jindal and Zihan Li and Aaron Lin and Kyle Schmaus and Jay Tayade and Weicheng Zhao and Anupam Datta and Nathan Wiegand and Dimitris Tsirogiannis},
  year={2025},
  eprint={2511.07663},
  archivePrefix={arXiv},
  primaryClass={cs.DB}
}

@misc{lao2026sembenchbenchmarksemanticquery,
      title={SemBench: A Benchmark for Semantic Query Processing Engines}, 
      author={Jiale Lao and Andreas Zimmerer and Olga Ovcharenko and Tianji Cong and Matthew Russo and Gerardo Vitagliano and Michael Cochez and Fatma Özcan and Gautam Gupta and Thibaud Hottelier and H. V. Jagadish and Kris Kissel and Sebastian Schelter and Andreas Kipf and Immanuel Trummer},
      year={2026},
      eprint={2511.01716},
      archivePrefix={arXiv},
      primaryClass={cs.DB},
      url={https://arxiv.org/abs/2511.01716}, 
}

@inproceedings{raasveldt2019duckdb,
  title={DuckDB: an Embeddable Analytical Database},
  author={Raasveldt, Mark and M{\"u}hleisen, Hannes},
  booktitle={Proceedings of the 2019 International Conference on Management of Data (SIGMOD)},
  pages={1981--1984},
  year={2019},
  publisher={ACM}
}

@mastersthesis{ebergen2023join,
  title={Join Order Optimization with (Almost) No Statistics},
  author={Ebergen, Tom},
  school={Vrije Universiteit Amsterdam},
  year={2022}
}

@inproceedings{hellerstein1993predicate,
  title={Predicate Migration: Optimizing Queries with Expensive Predicates},
  author={Hellerstein, Joseph M. and Stonebraker, Michael},
  booktitle={Proceedings of the 1993 ACM SIGMOD International Conference on Management of Data},
  pages={267--276},
  year={1993},
  publisher={ACM}
}

@article{liu2025optimizing,
  title={Optimizing llm queries in relational data analytics workloads},
  author={Liu, Shu and Biswal, Asim and Kamsetty, Amog and Cheng, Audrey and Schroeder, Luis G and Patel, Liana and Cao, Shiyi and Mo, Xiangxi and Stoica, Ion and Gonzalez, Joseph E and others},
  journal={Proceedings of Machine Learning and Systems},
  volume={7},
  year={2025}
}

@article{chaudhuri1999optimization,
  title={Optimization of Queries with User-Defined Predicates},
  author={Chaudhuri, Surajit and Shim, Kyuseok},
  journal={ACM Transactions on Database Systems},
  volume={24},
  number={2},
  pages={177--228},
  year={1999},
  publisher={ACM}
}

@article{ma2026can,
  title={Can AI Agents Answer Your Data Questions? A Benchmark for Data Agents},
  author={Ma, Ruiying and Shankar, Shreya and Chen, Ruiqi and Lin, Yiming and Zeighami, Sepanta and Ghosh, Rajoshi and Gupta, Abhinav and Gupta, Anushrut and Gopal, Tanmai and Parameswaran, Aditya G},
  journal={arXiv preprint arXiv:2603.20576},
  year={2026}
}

@misc{tpch2024,
  title={{TPC-H} Benchmark Specification},
  author={{Transaction Processing Performance Council}},
  year={2024},
  howpublished={\url{https://www.tpc.org/tpch/}},
  note={Accessed: 2026-03-31}
}

@inproceedings{wehrstein2025graceful,
  title={GRACEFUL: A Learned Cost Estimator for UDFs},
  author={Wehrstein, Johannes and Bang, Tiemo and Heinrich, Roman and Binnig, Carsten},
  booktitle={2025 IEEE 41st International Conference on Data Engineering (ICDE)},
  pages={2450--2463},
  year={2025},
  organization={IEEE}
}

@article{zhu2025beyond,
  title={Beyond Relational: Semantic-Aware Multi-Modal Analytics with LLM-Native Query Optimization},
  author={Zhu, Junhao and Chen, Lu and Ke, Xiangyu and Fang, Ziquan and Li, Tianyi and Gao, Yunjun and Jensen, Christian S},
  journal={arXiv preprint arXiv:2511.19830},
  year={2025}
}

@article{satriani2025logical,
  title={Logical and physical optimizations for sql query execution over large language models},
  author={Satriani, Dario and Veltri, Enzo and Santoro, Donatello and Rosato, Sara and Varriale, Simone and Papotti, Paolo},
  journal={Proceedings of the ACM on Management of Data},
  volume={3},
  number={3},
  pages={1--28},
  year={2025},
  publisher={ACM New York, NY, USA}
}

@article{chasialis2025optimizing,
  title={Optimizing UDF Queries in SQL Data Engines},
  author={Chasialis, Konstantinos and Foufoulas, Yannis and Simitsis, Alkis and Ioannidis, Yannis},
  year={2025}
}

@article{zhang2025scaledoc,
  title={ScaleDoc: Scaling LLM-based Predicates over Large Document Collections},
  author={Zhang, Hengrui and Hui, Yulong and Liu, Yihao and Zhang, Huanchen},
  journal={arXiv preprint arXiv:2509.12610},
  year={2025}
}

@article{mang2026argus,
  title={Automated Discovery of Test Oracles for Database Management Systems Using LLMs},
  author={Mang, Qiuyang and He, Runyuan and Zhong, Suyang and Liu, Xiaoxuan and Zhang, Huanchen and Cheung, Alvin},
  journal={Proceedings of the ACM on Management of Data},
  volume={4},
  number={3},
  pages={1--40},
  year={2026},
  publisher={ACM}
}

@article{giannakouris2025lambda,
  title={$\lambda$-Tune: Harnessing Large Language Models for Automated Database System Tuning},
  author={Giannakouris, Victor and Trummer, Immanuel},
  journal={Proceedings of the ACM on Management of Data},
  volume={3},
  number={1},
  pages={1--26},
  year={2025},
  publisher={ACM}
}

@article{hong2024text2sql,
  title={Next-Generation Database Interfaces: A Survey of LLM-based Text-to-SQL},
  author={Hong, Zijin and Yuan, Zheng and Zhang, Qinggang and Chen, Hao and Dong, Junnan and Huang, Feiran and Huang, Xiao},
  journal={arXiv preprint arXiv:2406.08426},
  year={2024}
}

@article{li2024llmr2,
  title={LLM-R2: A Large Language Model Enhanced Rule-based Rewrite System for Boosting Query Efficiency},
  author={Li, Zhaodongshui and Gao, Haitao and Wang, Huiming and Cong, Gao and Bing, Lidong},
  journal={arXiv preprint arXiv:2404.12872},
  year={2024}
}

@article{tan2025llmqo,
  title={Can Large Language Models Be Query Optimizer for Relational Databases?},
  author={Tan, Jie and Zhao, Kangfei and Rui, Li and Xu, Jeff and Yu, Chengzhi and Piao, Hong and Cheng, Helen and Meng, Deli and Zhao, Yu and Rong, Yu},
  journal={arXiv preprint arXiv:2502.05562},
  year={2025}
}

@article{zhou2025cracksql,
  title={CrackSQL: A Hybrid SQL Dialect Translation System Powered by Large Language Models},
  author={Zhou, Wei and Gao, Yuyang and Zhou, Xuanhe and Li, Guoliang},
  journal={arXiv preprint arXiv:2504.00882},
  year={2025}
}

@article{zhou2024dbgpt,
  title={Db-gpt: Large Language Model Meets Database},
  author={Zhou, Xuanhe and Sun, Zhaoyan and Li, Guoliang},
  journal={Data Science and Engineering},
  volume={9},
  number={1},
  pages={102--111},
  year={2024}
}

@article{arch2024key,
  title={The key to effective udf optimization: Before inlining, first perform outlining},
  author={Arch, Samuel and Liu, Yuchen and Mowry, Todd C and Patel, Jignesh M and Pavlo, Andrew},
  journal={Proceedings of the VLDB Endowment},
  volume={18},
  number={1},
  year={2024},
  publisher={VLDB}
}

@inproceedings{franz2024dear,
  title={Dear User-Defined Functions, Inlining isn't working out so great for us. Let's try batching to make our relationship work. Sincerely, SQL.},
  author={Franz, Kai and Arch, Samuel and Hirn, Denis and Grust, Torsten and Mowry, Todd C and Pavlo, Andrew},
  booktitle={CIDR},
  year={2024}
}

@article{stoian2024dpconv,
  title={DPconv: Super-Polynomially Faster Join Ordering},
  author={Stoian, Mihail and Kipf, Andreas},
  journal={Proceedings of the ACM on Management of Data},
  volume={2},
  number={6},
  pages={1--26},
  year={2024},
  publisher={ACM New York, NY, USA}
}

\end{document}